\documentclass[final]{siamltex}
\usepackage{bm,amsmath}
\usepackage{algorithm}
\usepackage{algpseudocode}
\usepackage{amsfonts}
\usepackage{graphicx}
\usepackage{enumitem}%
\usepackage{mathrsfs}

\DeclareMathAlphabet\mathpzc{OT1}{pzc}{m}{it}
\let\mathcal=\mathpzc

\def\intinfty{\int\limits_{\!\!-\infty\,\,}^{\,\,\infty\!\!}\kern-0.0em}
\def\iintinfty{\mathop{\int\!\!\int}\limits_{\!\!-\infty\,\,}^{\,\,\infty\!\!}\kern-0.0em}
\def\iiintinfty{\mathop{\int\!\!\int\!\!\int}\limits_{\!\!-\infty\,\,}^{\,\,\infty\!\!}\kern-0.0em}

\let\<=\langle
\let\>=\rangle
\let\^=\hat
\def\~#1{{\-ox{\sf#1}}}
\def\N{{\mathbb N}}

\let\-=\mathbf

\def\@#1{{\cal #1}}

\title{A hybrid adaptive MCMC algorithm in function spaces\thanks{This work was
supported by the NSFC under grant number 11301337.}} 

\author{Qingping Zhou\footnotemark[2] \and Zixi Hu\footnotemark[3] \and Zhewei Yao\footnotemark[3] \and Jinglai Li\footnotemark[4]}

\begin{document}
\maketitle
\newcommand{\slugmaster}{%
\slugger{juq}{xxxx}{xx}{x}{x--x}}
\renewcommand{\thefootnote}{\fnsymbol{footnote}}

\footnotetext[2]{Department of Mathematics and Institute of Natural Sciences,   
Shanghai Jiao Tong University, 800 Dongchuan Rd, Shanghai 200240, China, (zhouqingping@sjtu.edu.cn). }

\footnotetext[3]{Department of Mathematics, Zhiyuan College,   
Shanghai Jiao Tong University, 800 Dongchuan Rd, Shanghai 200240, China, (\{yzw,hzx\}@sjtu.edu.cn).}
\footnotetext[4]{Corresponding author, Institute of Natural Sciences, Department of Mathematics, and 
the MOE Key Laboratory of Scientific and Engineering Computing, 
Shanghai Jiao Tong University, 800 Dongchuan Rd, Shanghai 200240, China, (jinglaili@sjtu.edu.cn).}

\renewcommand{\thefootnote}{\arabic{footnote}}

\begin{abstract}
The preconditioned Crank-Nicolson (pCN) method is a Markov Chain Monte Carlo (MCMC) scheme, specifically designed to perform Bayesian inferences in function spaces. 
Unlike many standard MCMC algorithms, the pCN method can preserve the sampling efficiency under the mesh refinement, a property referred to as being dimension independent.  
In this work we consider an adaptive strategy to further improve the efficiency of pCN. 
In particular we develop a hybrid adaptive MCMC method: the algorithm performs an adaptive Metropolis scheme in a chosen finite dimensional subspace,
and a standard pCN algorithm in the complement space of the chosen subspace.  We show that the proposed algorithm satisfies certain important ergodicity conditions.  
Finally with numerical examples we demonstrate that the proposed method has competitive performance with existing adaptive algorithms.  
\end{abstract}


\begin{AMS}\end{AMS}

\pagestyle{myheadings}
\thispagestyle{plain}
\markboth{Q. Zhou, Z. Hu, Z. Yao and J. Li}{Hybrid MCMC in function spaces}

\section{Introduction}
Many real-world inverse problems require to 
estimate unknowns that are functions of space and/or time. 
Solving such problems with the Bayesian approaches~\cite{kaipio2006statistical,stuart2010inverse}, 
has become increasing popular, largely due to its ability to quantify the uncertainty in the estimation results. 
To implement the Bayesian inference in those problems, it is often required to perform Markov Chain Monte Carlo (MCMC) simulations in function spaces.
Usually, the unknown is represented with a finite-dimensional parametrization, and then MCMC is applied to the resulting finite dimensional problems. 
Many standard MCMC algorithms, such as the popular random walk Metropolis-Hastings (RWMH), are known to be dimension dependent, 
as they can become arbitrarily slow as the discretization dimensionality increases~\cite{roberts2001optimal,mattingly2012diffusion}. 
To this end, a very interesting line of research is to develop \emph{dimension-independent} MCMC algorithms 
by requiring the algorithms to be well-defined in the function spaces.
In particular, a family of dimension-independent MCMC algorithms, known as the preconditioned Crank Nicolson (pCN) algorithms, were presented in \cite{cotter2013mcmc} by constructing a Crank-Nicolson 
discretization of a stochastic partial differential equation (SPDE) that preserves the reference measure. 
Several variants of the pCN algorithms have been developed to further improve the sampling efficiency. 
For example, a class of algorithms accelerate the pCN scheme using the gradient information of the likelihood functions, such algorithms include, 
 the operator-weighted proposal method~\cite{law2014proposals},  the dimension-independent likelihood-informed MCMC~\cite{cui2016dimension}, 
and  the generalized pCN algorithm~\cite{rudolf2015generalization},  just to name a few. 

In this work, we focus on an alternative way to improve the sampling efficiency, the adaptive MCMC methods. 
Simply speaking, the adaptive MCMC algorithms improve the proposal based on the sampling history from the targeting distribution~(c.f. \cite{andrieu2008tutorial,atchade2009adaptive,roberts2009examples} and the references therein) as the iterations proceed. A major advantage of the adaptive methods is that they only require the ability
to evaluate the likelihood functions, which makes them particularly convenient for problems 
with black-box models. In a recent work~\cite{hu2015adaptive}, we developed 
an adaptive pCN (ApCN) algorithm based on the idea of adapting the proposal covariance to approximate that of the posterior. 
The ApCN algorithm requires the proposal covariance to be diagonal, assuming the unknown is represented with the Karhunen-Lo\`eve expansion~\cite{hu2015adaptive}, as that the implementation involves computing the square root of a large matrix,
which is very intensive if the covariance is not diagonal~\cite{hu2015adaptive}. 
In this work, we present an improved adaptive MCMC algorithm for functions, particularly addressing this limitation of the ApCN algorithm. 
The proposed algorithm is essentially a hybrid scheme:
it performs an adaptive Metropolis (AM) scheme in a chosen finite dimensional subspace of the state space and a pCN in the complement  of it. 
In the present algorithm, the proposal covariance directly approximates that of the posterior without assuming a diagonal structure.
With numerical examples, we show that the present algorithm can outperform the ApCN scheme, thanks to the relaxation of the diagonal structure.
Another important improvement of the present method is about the convergence property of the algorithm.
Recall that, to show the ergodicity property of the ApCN method, we need to impose an artificial modification of the likelihood function to ensure that the support
of the posterior is bounded~\cite{hu2015adaptive}; however, for the present hybrid algorithm, we can show that it satisfies the same ergodicity conditions without modifying the likelihood function. 

We note that, other dimension independent MCMC algorithms are available, such as
 the stochastic Newton MCMC~\cite{martin2012stochastic,petra2014computational},
the infinite dimensional  Riemann manifold Hamiltonian Monte Carlo method~\cite{bui2014solving}, 
 the dimension independent adaptive Metropolis~\cite{chen2015accelerated}, 
 and the infinite dimensional independence sampler~\cite{feng2015adaptive}.  
Comparison of these methods with the pCN based algorithms are not in the scope of the work.

The reminder of the paper is organized as follows. In section~\ref{sec:method} we present our hybrid adaptive MCMC algorithm as well as 
some theoretical results regarding its ergodicity.  
In section~\ref{sec:examples} we provide several numerical examples to demonstrate the performance of the proposed algorithm. 
Finally we offer some concluding remarks in section~\ref{sec:conclusion}.

\section{The hybrid adaptive MCMC method}\label{sec:method}

\subsection{Bayesian inferences in function spaces}
We present the standard setup of the Bayesian inverse problem following \cite{stuart2010inverse}. 
We consider a separable {Hilbert} space $X$ with inner product $\<\cdot,\cdot\>_X$.
 Our goal is to estimate the unknown  $u\in X$ from data $y\in Y$ where $Y$ is the data space and $y$ is related to $u$ via a likelihood function $L(x,y)$.
In the Bayesian inference we assume that the prior $\mu_0$ of $u$, is a  (without loss of generality)~zero-mean Gaussian measure defined on $X$ with covariance operator $\@C_0$,
i.e. $\mu_0 = N(0,\@C_0)$. 
Note that $\@C_0$ is symmetric positive and of trace class.
The range of $\@C_0^{\frac12}$,
\[E = \{u = \@C_0^{\frac12} x\, |\, x\in X\}\subset X,\]
which is a Hilbert space equipped with inner product~\cite{da2006introduction},
\[\<\cdot,\cdot\>_E = \<\@C_0^{-\frac12}\cdot,\@C_0^{-\frac12}\cdot\>_X ,\]
is called the Cameron-Martin space of measure $\mu_0$. 
In this setting, the posterior measure $\mu^y$ of $u$ conditional on data $y$
is provided by the Radon-Nikodym derivative:
\begin{equation} \frac{d\mu^y}{d\mu_{0}}(u) =\frac1Z\exp(-\Phi^y(u)),
 \label{e:bayes}
\end{equation}
with $Z$ being a normalization constant, 
which can be interpreted as the Bayes' rule in the infinite dimensional setting.
In what follows, without causing any ambiguity, we shall drop the superscript $y$ in $\Phi^y$ and $\mu^y$ for simplicity, while keeping 
in mind that these items depend on the data $y$.  
A typical example is the so-called Bayesian inverse problems~\cite{kaipio2006statistical,stuart2010inverse}, which assumes that the unknown $u$ is mapped to the data $y$ via a forward model  
$y = G(u)+\zeta$,
where $G:X\rightarrow R^d$ and $\zeta$ is a $d$-dimensional Gaussian noise with mean zero and covariance $C_\zeta$. 
In this case $\Phi(u) = |C_\zeta^{-\frac12}(Gu-y)|^2_2$.

For the inference problem to be well-posed, one typically requires the functional $\Phi$ to satisfy the Assumptions (6.1) in \cite{cotter2013mcmc}.
Finally we quote the following lemma~(\cite{da2006introduction}, Chapter~1), which will be useful later:
\begin{lemma}
There exists a complete orthonormal \label{lm:eigens}
basis $\{e_j\}_{j\in\N}$ on $X$ and a sequence of non-negative numbers $\{\alpha_j\}_{j\in\N}$
such that ${\@C_0} e_j = \alpha_j e_j$ and $\sum_{j=1}^\infty \alpha_j <\infty$, i.e., 
 $\{e_j\}_{k\in\N}$ and $\{\alpha_j\}_{k\in\N}$ being the eigenfunctions and eigenvalues of $\@C_0$ respectively.
\end{lemma}\\
For convenience's sake, we assume that the eigenvalues are in a descending order: $\lambda_j\geq\lambda_{j+1}$ for any $j\in\N$.
$\{e_j\}_{j=1}^\infty$ are known as the Karhunen-Lo\`eve~(KL) modes associated with $\@N(0,\@C_0)$.

\subsection{The preconditioned Crank-Nicolson algorithm}
We now briefly review the family of Crank-Nicolson algorithms for infinite dimensional Bayesian inferences, following the presentation of \cite{cotter2013mcmc}. 
Simply speaking the algorithms are based on the stochastic partial differential equation (SPDE)
\begin{equation}
\frac{du}{ds}=-\mathcal{K}\mathcal{L}u+\sqrt{2\mathcal{K}}\frac{db}{ds},\label{e:spde}
\end{equation}
where $\@L=\@C_0^{-1}$ is the precision operator for $\mu_0$, $\@K$ is a positive operator,
 and $b$ is a Brownian motion in $X$ with covariance operator the identity. 
The proposal is then derived by applying the Crank-Nicolson (CN) discretization to the SPDE~\eqref{e:spde}, yielding, 
\begin{equation}
v=u-\frac12\delta\mathcal{K}\mathcal{L}(u+v)+\sqrt{2\mathcal{K}\delta}\xi_0, \label{e:prop}
\end{equation}
for a white noise $\xi_0$ and $\delta\in(0,2)$.
In \cite{cotter2013mcmc}, two choices of $\@K$ are proposed, resulting in two different algorithms.
First, one can choose $\@K=\@I$, the identity, obtaining: 
\[
(2\@C+\delta \@I)v=(2\@C-\delta \@I)u+\sqrt{8\delta}w,
\]
where $w\sim \@N(0,\@C_0)$,
which is known as the plain CN algorithm. 
Alternatively one can choose $\@K=\@C_0$, resulting in the pCN proposal:
\begin{equation}
v=(1-\beta^2)^{\frac12}u+ \beta w, \label{e:pcn}
\end{equation}
where \[\beta = \frac{\sqrt{8\delta}}{2+\delta}.\]
It is easy to see that $\beta\in[0,1]$.
In both CN and pCN algorithms,  the acceptance probability is
\begin{equation}
a(v,u) = \min\{1, \exp{\Phi(u)-\Phi(v)}\}. \label{e:acc}
\end{equation}

\subsection{The hybrid  algorithm}\label{pr:acc}
We start with a non adaptive version of the proposed hybrid algorithm. 
For a prescribed integer $J>0$ (the interpretation of $J$ and how to determine it will be discussed later), we let $X^+ = \mathrm{span}\{e_1,...,e_J\}$ and ${X}^- = X\backslash X^+$. 
For any $u\in X$, we can write 
$u=u^+ +u^-$ where $u^+\in X^+$ and $u^-\in X^-$. 
Our algorithm proposes according to  
\begin{subequations} \label{e:propnew}
\begin{equation}
v =v^+ +v^-,
\quad \mathrm{with}\quad  v^+= u^+ +\beta w^+,\quad \mathrm{and}\quad v^- = (1-\beta^2)^{\frac12} u^- +\beta w^-,
\end{equation}
where 
\begin{equation}
w^+ = \sum_{i=1}^J w_i e_i, \quad \mathrm{with} \quad (w_1,...,w_J)^T\sim N(0,\Sigma) ,
\end{equation}
with $\Sigma$ being a $J\times J$ covariance matrix (and thus it must be symmetric and positive definite), and
 \begin{equation}
 w^- = \sqrt{\@B} \xi_0, 
\end{equation}
in which
\[\@B \,\cdot = \sum_{j=J+1}^\infty \alpha_{j}\<e_j,\cdot\>e_j,\]
\end{subequations}
and $\xi_0$ is a white Gaussian noise.
It is easy to see that the proposal defined by Eqs.~\eqref{e:propnew} is a Gaussian measure $\@N(m(u),\beta^2\@C)$ with mean $m = u^++(1-\beta^2)^{\frac12} u^-$ and covariance $\@C$ such that
\[
\@C\cdot= (\<e_1,\cdot\>,...\<e_J,\cdot\>) \Sigma (e_1,...,e_J)^T + \@B\cdot.
\]
The key in the algorithm is to choose an appropriate covariance matrix $\Sigma$.
Before discussing how to choose $\Sigma$, we first  show that under mild conditions, the proposal~\eqref{e:propnew} results in well-defined acceptance probability in a function space:
\begin{proposition}
Let $q(u,\cdot)$ be the proposal distribution associated to Eq.~\eqref{e:propnew}. 
Define measures $\eta(du,dv)=q(u,dv)\mu(du)$ and $\eta^\bot(du,dv)=q(v,du)\mu(dv)$ on $X\times X$. If $\Sigma$ is symmetric and positive definite,  $\eta^\bot$ is absolutely continuous with respect to $\eta$, and 
\begin{equation}
\frac{d\eta^\bot}{d\eta}(u,v) = \exp[\Phi(u)-\Phi(v) +\frac12 \sum_{i=1}^J \frac{(\<u,\,e_i\>^2-\<v,e_i\>^2)}{\alpha_i}]. \label{e:accnew}
\end{equation}
\end{proposition}
\begin{proof}
Define $\eta_0(du,dv)=q(u,dv)\mu_0(du)$ and $\eta_0^\bot(du,dv)=q(v,du)\mu_0(dv)$. Both $\eta_0$ and $\eta_0^\bot$ are Gaussian. First, we have 
\[
\eta(du,dv)=q(u,dv)\mu(du),\quad
\eta_0(du,dv)=q(u,dv)\mu_0(du),
\]
and $\mu$,$\mu_0$ are equivalent. It follows that $\eta$ and $\eta_0$ are equivalent and 
\begin{equation}
\frac{d\eta}{d\eta_0}(u,v)=\frac{d\mu}{d\mu_0}(u)=  \frac1{Z}\exp(-\Phi(u)). \label{e:eta}
\end{equation}
Obviously following the same argument, we also have that $\eta^\bot$ and $\eta_0^\bot$ are equivalent and 
\begin{equation}
\frac{d\eta^\bot}{d\eta_0^\bot}(u,v)= \frac1Z\exp(-\Phi(v)). \label{e:etabot}
\end{equation}
By some elementary calculations, one can show that, 
\begin{equation}
\frac{d\eta_0^\bot}{d\eta_0}(u,v)= \exp(\frac12\sum_{i=1}^J\frac{(\<u,e_i\>^2-\<v,e_i\>^2)}{\alpha_i}). \label{e:eta0}
\end{equation}
It follows immediately from Eqs.~\eqref{e:eta} - \eqref{e:eta0} that $\eta$ and $\eta^\bot$ are equivalent and 
Eq.~\eqref{e:accnew} holds. 
\end{proof}

From the detailed balance condition one can derive that the acceptance probability of proposal~\eqref{e:propnew} is 
\begin{equation}
a(u,v) =\min\{1, \frac{d\eta^\bot}{d\eta}(u,v)\}, \label{e:accnew2}
\end{equation}
where $\frac{d\eta^\bot}{d\eta}(u,v)$ is given by Eq.~\eqref{e:accnew}.

We now consider how to determine $\Sigma$. 
A rule of thumb in designing efficient MCMC algorithms is that the proposal covariance should be close to the
covariance operator of the posterior~\cite{roberts2001optimal,haario2001adaptive}. 
Now suppose the posterior covariance is $\@C^{y}$, and
one can determine the proposal covariance $\@C$ by solving
  \begin{equation}
  \min_{\Sigma}\|\@C-\@C^y\|_{HS}, \label{e:minhs}
  \end{equation}
   where $\|\cdot\|_{HS}$ is the Hilbert-Schmidt operator norm.
 By some basic algebra, we can show that the optimal solution of Eq~\eqref{e:minhs} is
 \[\sigma_{i,j} = \< \@C^y e_i,e_j\>^{-1} ,\]
 for $i,j=1...J$, where $\sigma_{i,j}$ are the entries of $\Sigma$. 
Since $\@C^y$ is the posterior covariance,  for any $v$ and $v' \in X$, we have~\cite{da2006introduction},
 \begin{equation}
\<\@C^y v, v'\> =\int \<v,u-m^y\>\<v',u-m^y\> \mu(du), \label{e:cov}
\end{equation}
where $m^y$ is the mean of $\mu$.
Using Eq.~\eqref{e:cov}, we can derive that
\begin{equation}
 \sigma_{i,j} =\int\<u-m^y,e_i\>\<u-m^y,e_j\>\mu(du), \label{e:sigma_ij} 
\end{equation}
 for $i,j = 1...J$.

Since Eq.~\eqref{e:sigma_ij} can not be computed directly, we estimate the covariance matrix $\Sigma$ with the adaptive Metropolis method. 
Simply speaking, the AM algorithm starts with an initial guess of $\Sigma$ and then adaptively updates the $\Sigma$ based on the sample history.
Namely,  suppose we have a set of samples $\{ u_1,...,u_n\}$ and  let $x_i$ be the projection of $u_i$ onto the basis $(e_1,...,e_J)$:
\[ x_i = (\<u_i,e_1\>,...\<u_i,e_J\>).\]
 We estimate $\Sigma$ with
\begin{subequations}\label{e:update}
\begin{gather}
\hat{x} = \frac1n\sum_{i=1}^n x_i,\\
\hat{\Sigma} = \frac{1}{n-1}\sum_{i=1}^n (x_i-\hat{x}) (x_i-\hat{x})^T + \delta I, \label{e:sighat}
\end{gather}
\end{subequations}
where $\delta$ is a small positive constant and $I$ is the identity matrix. 
Note that the term $\delta I$ in Eq.~\eqref{e:sighat} is introduced to stabilize the iteration, as is used in \cite{haario2001adaptive}. 
For efficiency's sake,  Eq~\eqref{e:update} can be recast in a recursive form (Eq.~(7) in \cite{haario2001adaptive}). 
It should be noted that it is not robust to estimate the parameter values with a very small number of samples, 
and to this end we employ a pre-run,  drawing a certain number of samples with a standard pCN algorithm, before starting the adaptation.  
Finally, we note that, in principle, a sample $x$ with very large norm can distort the estimate of the covariance matrix $\Sigma$,   
and to prevent this from happening, we introduce a norm threshold $R\gg0$, and if a sample’s norm exceeds this threshold, it is not use it to update the covariance. 
This step is essential for our convergence results. 
We describe the complete hybrid adaptive algorithm in Algorithm~\ref{al:apcn}. 
\begin{algorithm}[!tb]
 \caption{The hybrid adaptive algorithm}
    \label{al:apcn}
    \begin{algorithmic}[1]
    
\State  Initialize $u_1\in X $;
\State drawn $N'$ samples with a standard pCN algorithm, denoted as $\{u_i\}_{i=1}^{N'}$;
\State Let $S = \{u_i, i =1...N'~|~\|u_i\|_X<R\}$
\State Compute $\Sigma$ using Eqs.~\eqref{e:update} and samples in $S$;

   \For {$n=N'$ to $N-1$}
      
	\State Propose $v$ using Eq~\eqref{e:propnew};	
        \State Draw $\theta\sim U[0,1]$
        \State Compute $a(u,v)$ with Eq.~\eqref{e:accnew2}; 
        
        \If {$\theta\leq a$ } 
            \State $u_{n+1}\leftarrow v$;
						\Else
						\State $u_{n+1}\leftarrow u^{n}$;
        \EndIf
				\If{$\|u^{n+1}\|_X<R$}
				\State $S\leftarrow S\cup\{u^{n+1}\}$;	
				\State Update $\Sigma$ using Eqs.~\eqref{e:update} and samples in $S$;
		    \EndIf	
						\EndFor

    \end{algorithmic}
		\medskip
		
    \end{algorithm}

The basis idea behind the proposed method may become 
 more clear if  we look at the projections of the proposal onto each eigenmodes:
\begin{equation}\label{e:projection}
\<v, e_i\> =\begin{cases} \<u, e_i\>+\beta w_i &\mbox{for } i \leq J, \\
(1-\beta^2)^{\frac12}\<u,e_i\>+\beta w_i & \mbox{for } i >J, \end{cases}
\end{equation}
where 
$(w_1,...,w_J)^T\sim N(0,\Sigma)$ 
and $w_i \sim N(0,\alpha_i)$ for $i>J$. 
Eq~\eqref{e:projection} shows the \emph{hybrid} nature of the algorithm:
it performs an AM algorithm in a finite dimensional space spanned by $\{e_1,...,e_J\}$ with the proposal covariance adapted 
to approximate that of the posterior,
and a standard pCN sampler for all $j>J$.
The intuition behind our algorithm is based on the assumption that the (finite-resolution) data is only informative about a finite number 
of KL modes of the prior.  In particular, the data can not provide information about the modes that are highly oscillating (associated with small eigenvalues) and for those modes, the posterior is 
approximately the prior. 
In this case, in the finite dimensional subspace spanned by the
modes that are significantly informed by the data, we shall perform an AM algorithm to improve the sampling efficiency; 
in its complement space, we just use the standard pCN method to preserve the dimensional independence of the MCMC scheme.

Finally an important issue in the implementation is to determine the number of adapted eigenvalues $J$.
Following \cite{hu2015adaptive}, we let $J =\min\{j\in \N\}$ such that, 
\[\frac{\sum_{i=1}^j{\alpha_j}}{\sum_{i=1}^\infty{\alpha_i}}>\rho,\]
where $0<\rho<1$ is a prescribed number (e.g. $\rho=0.9$).
In Section~\ref{sec:examples}, with numerical examples, we demonstrate how the choice of $J$ affects the sampling efficiency of the algorithm. 
\subsection{The convergence property}
It is well known that, the chain constructed with an adaptive MCMC algorithm may not converge to the target distribution, i.e., losing its ergodicity. 
Thus, for a new adaptive algorithm, it is important to study whether it can correctly converge to the target distribution.   
It has been proved by Roberts and Rosenthal \cite{roberts2009examples} that, an adaptive MCMC algorithm
has the correct asymptotic convergence, provided that it satisfies 
the Diminishing  Adaptation~(DA) condition, which, loosely speaking, 
requires the transition probabilities to converge as the iteration proceeds,  
 and the Containment condition.
It has also been suggested by the authors that the Containment condition is often merely a technical condition which is satisfied for virtually all reasonable adaptive schemes~\cite{roberts2009examples},
and thus  here we show that  the proposed hybrid algorithm satisfies  the DA condition.  
Recall that, to show the ApCN algorithm satisfies the DA condition, the likelihood function is modified to be
\[ \frac{d\mu^y}{d\mu_{0}}(u) \propto
\begin{cases}
\exp(-\Phi(z)), &\|u\|_X\leq R_{\max}\cr
0,&\|u\|_X> R_{\max},
\end{cases}
\]
where $R_{\max}$ is a prescribed positive constant. 
Removal of this artificial modification is certain desirable, and in what follows we shall show that the present hybrid algorithm satisfies the DA condition, without making such a modification. 

Suppose at iteration $n$, we have samples $\{u_0,\,u_1,\,\cdots,\,u_{n-2},\,u\}$ and for simplicity we define the notation: 
$\zeta_{n-2}=(u_0,u_1,\cdots,u_{n-2})$. 
Let $\Sigma_{n}$ be the subspace covariance matrix estimated with $\{u_0,\,u_1,\,\cdots,\,u_{n-2},\,u\}$ using Eq.~\eqref{e:update},   
and $\@C_{n,\zeta_{n-2}}(u)$ be the corresponding proposal covariance operator. 
We define $q_{n,\zeta_{n-2}}(u;dv)=\@N(m(u),\beta^2\@C_{n,\zeta_{n-2}}(u))$, i.e., the proposal distribution at iteration $n$, and
\[
Q_{n,\zeta_{n-2}}(u,dv)=a(u,v)q_{n,\zeta_{n-2}}(u,dv)+\delta_u(dv)(1-\int a(u,v')q_{n,\zeta_{n-2}}(u,dv'))
\] 
where $a(\cdot,\cdot)$ is given by Eq.~\eqref{e:accnew2}.
We then have the following theorem (the DA condition): 
\begin{theorem} \label{th:da}
 There is a fixed positive constant $\gamma$ such that 
\[
\sup_{u\in X}\|Q_{n,\zeta_{n-2}}(u,\cdot)-Q_{n+1,\zeta_{n-1}}(u,\cdot)\|\leq\frac{\gamma}{n}
\]
for any $\zeta_{n-1}$ and $\zeta_{n-2}$ such that $\zeta_{n-1}$ is a direct continuation of $\zeta_{n-2}$.  
Here  $\|\cdot\|$ is the total variation norm.
\end{theorem}


\begin{proof}
First it is easy to see that $q_{n,\zeta_{n-2}}(u;\cdot)$ and $q_{n+1,\zeta_{n-1}}(u;\cdot)$ are both Gaussian measures with same mean,
and we have
\begin{equation}
\frac{dq_{n,\zeta_{n-2}}(u;\cdot)}{dq_{n+1,\zeta_{n-1}}(u;\cdot)}(v) = \sqrt{\dfrac{|\Sigma_{n+1}|}{|\Sigma_{n}|}}\exp(\frac12\Delta x^T
(\Sigma_{n+1}^{-1}-\Sigma_n^{-1})\Delta x),
\end{equation}
where $\Delta x  = (\<v-u,e_1\>,...\<v-u,e_J\>)^T$. Let $A$ be any member of the $\sigma$-field of $X$, and we compute
\[
\begin{array}{ll}
&|Q_{n,\zeta_{n-2}}(u;A)-Q_{n+1,\zeta_{n-1}}(u;A)|\\
&=|\int_A a(u,v)q_{n,\zeta_{n-2}}(u;dv)+\delta_A(u)(1-\int_{X}{q_{n,\zeta_{n-2}}(u;dv')a(u,v')})\\
&-\int_A a(u,v)q_{n+1,\zeta_{n-1}}(u;dv)+\delta_A(u)(1-\int_{X}{q_{n+1,\zeta_{n-1}}(u;dv')a(u,v')})|\\
&\leq 2\int_X a(u,v)|\frac{dq_{n,\zeta_{n-2}}(u;\cdot)}{dq_{n+1,\zeta_{n-1}}(u;\cdot)}(v)-1|q_{n+1,\zeta_{n-1}}(u;dv)\\
&\leq 2\int_X |\frac{dq_{n,\zeta_{n-2}}(u;\cdot)}{dq_{n+1,\zeta_{n-1}}(u;\cdot)}(v)-1|q_{n+1,\zeta_{n-1}}(u;dv),\\
&= 2\int_X |\sqrt{\dfrac{|\Sigma_{n+1}|}{|\Sigma_{n}|}}
\exp(\frac12\Delta x^T(\Sigma_{n+1}^{-1}-\Sigma_n^{-1})\Delta x)-1|q_{n+1,\zeta_{n-1}}(u;dv),\\
&=\dfrac2{(2\pi)^{\frac{J}2}}\int_{\mathbb{R}^J}|\dfrac{1}{\sqrt{|\Sigma_{n}|}}\exp(-\frac12\Delta x^T\Sigma^{-1}_n\Delta x)-
\dfrac{1}{\sqrt{|\Sigma_{n+1}|}}\exp(-\frac12\Delta x^T\Sigma^{-1}_{n+1}\Delta x)|d\Delta x\\
&\leq c_1 \|\Sigma_{n} - \Sigma_{n+1}\|,
\end{array}
\]
for some constant $c_1>0$. 
If $\|u_{n}\|_X>R$, $\|\Sigma_n-\Sigma_{n-1}\|=0$; otherwise, following the same argument of the proof of Theorem 2 in \cite{haario2001adaptive}, 
we can show that $\|\Sigma_n-\Sigma_{n-1}\|\leq c_2/n $ for some constant $c_2>0$. It follows directly that the theorem holds. 
\end{proof}

\section{Numerical examples}\label{sec:examples}

\subsection{A Gaussian example}
Intuitively, we expect that the proposed hybrid method should be advantageous over ApCN in problems where the the eigenmodes are strongly correlated. 
To test this property, we construct a simple mathematical example. 
We assume the unknown is a function defined on the interval $[0,\,1]$, and the prior is taken to be a zero mean Gaussian with Mat\'ern covariance~\cite{rasmussen2006gaussian}:
\begin{equation}
K(t_1,t_2) = \sigma^2 \frac{2^{1-\nu}}{\mathrm{Gam}(\nu)}(\sqrt{2\nu}\frac{d}{l})^\nu B_\nu(\sqrt{2\nu}\frac{d}{l}), \label{e:matern}
\end{equation}
where $d=|t_1-t_2|$, $\mathrm{Gam}(\cdot)$ is the Gamma function, and $B_\nu(\cdot)$ is the modified Bessel function. 
A random function with the Mat\'ern covariance is $[\nu-1]$ mean square (MS) differentiable, and here we choose $\nu=5/2$ implying second order MS differentiability.
Moreover, we set $\sigma=1$ and $l=1$ in this example. 
We take the function $\Phi(u)$ to be 
\[ \Phi(u) = \frac12 x^T\,\Gamma \, x \]
where $x = (\<u,e_1\>,...\<u,e_K\>)^T$ for a positive integer $K$ and $\Gamma[i,j] = \exp(-(i-j)^2/\Delta)$ for $i,\,j = 1...K$ and a constant $\Delta>0$. 
In this example we choose $K=14$. 
It is easy to see that the posterior distribution is also Gaussian, and
by choosing different value of $\Delta$ we can control the posterior correlation between the eigenmodes. 
In particular, we perform numerical tests for the two cases: $\Delta=1$ (weak correlation) and $\Delta =14$ (strong correlation).  
In each case, we sample the posterior distribution with three methods: the standard pCN, ApCN, and the hybrid method. 
For the ApCN and the hybrid methods, we draw $5\times10^5$ samples with another $0.5\times10^5$ pCN samples used in the pre-run,
and for the pCN method, we directly draw $5.5\times10^5$ samples. Moreover, we set $J=14$ in both the ApCN and the hybrid methods.   
We note that, in all the numerical tests performed in this work, unless otherwisely stated, the unknown is represented with 201 grid points and the stepsize $\beta$ has been chosen in a way that the resulting acceptance probability is around $25\%$. 

We first show the results for $\Delta =1$.
In Fig.~\ref{f:acf_gauss1}, we plot the autocorrelation function (ACF) of the samples drawn by each method against the lag at $t=0.4$ and $t=0.8$. 
We then compute the ACF of lag $100$ at all the grid points, and show the results in Fig.~\ref{f:acf100-ess-gauss1} (left).  
The effective sample size (ESS) is another popular measure of the sampling efficiency of MCMC~\cite{Kass1998}, which gives an estimate of the number of effectively independent draws in the chain.
We compute the ESS per 100 samples of the unknown $u$ at each grid point and show the results in Fig.~\ref{f:acf100-ess-gauss1} (right).
We then show the same plots for $\Delta =14$ in Figs.~\ref{f:acf_gauss14} and \ref{f:acf100-ess-gauss14}. 
We can see from these plots that, in the weakly correlated case $\Delta=1$, the hybrid method is not clearly advantageous over the ApCN algorithm;
 in the strongly correlated case $\Delta=14$, the hybrid method performs much better than the ApCN algorithm, suggesting that
taking the covariances between eigenmodes into account can significantly improve the sampling efficiency in this case. 
These results agree well with our expectations.

\begin{figure}
\centerline{\includegraphics[width=.5\textwidth]{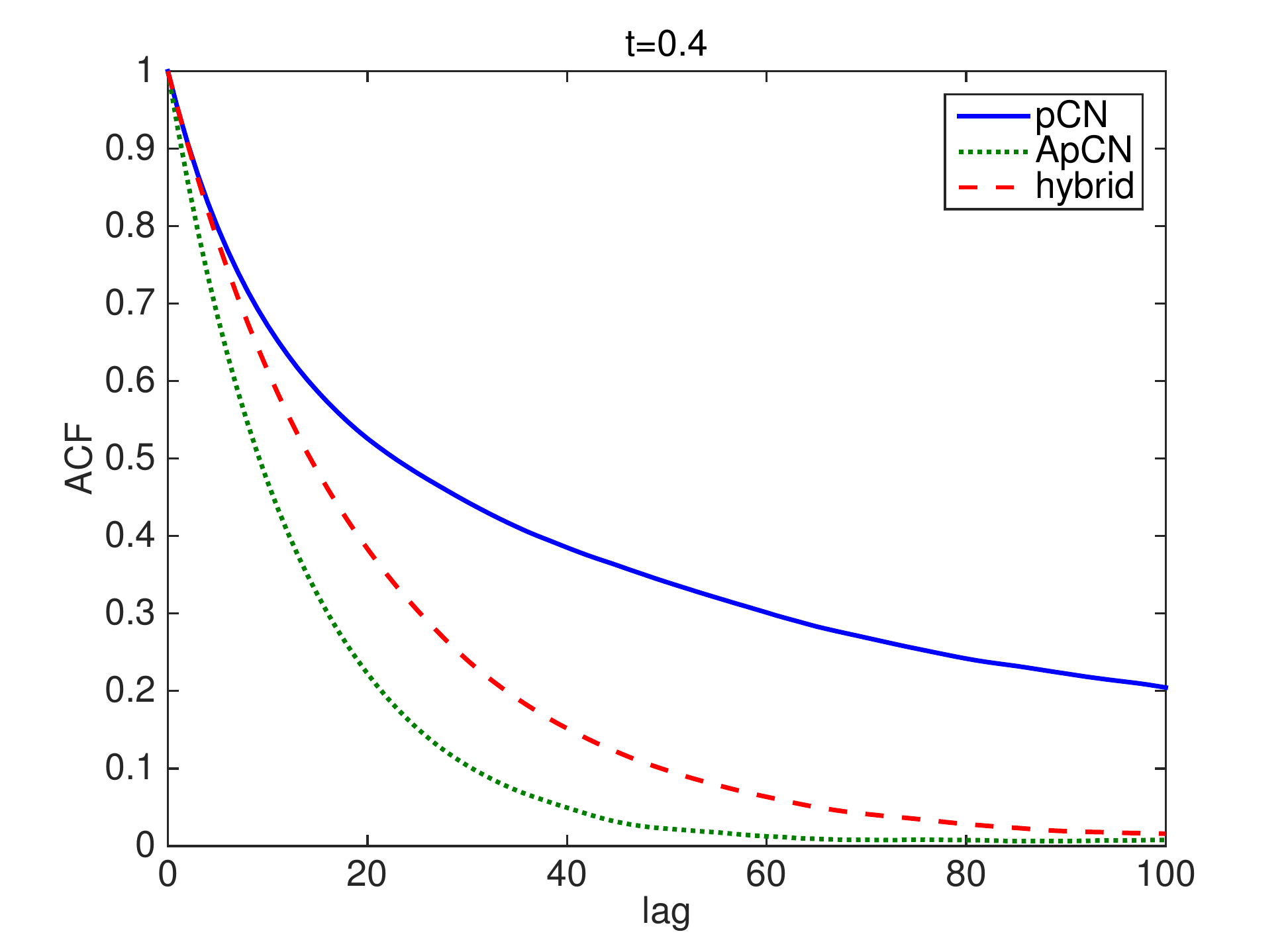}
\includegraphics[width=.5\textwidth]{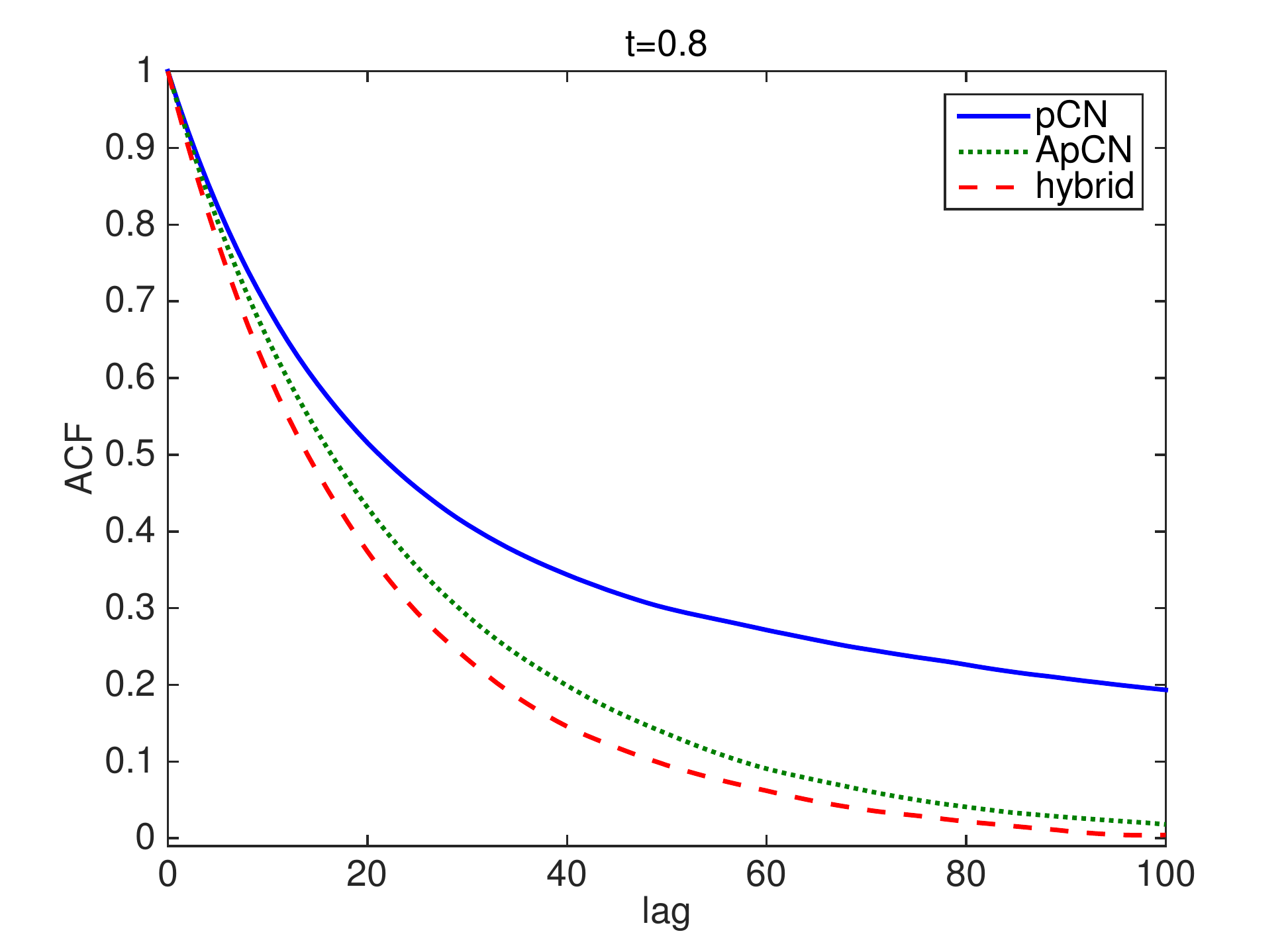}}
\caption{(for the Gaussian example: $\Delta=1$) ACF for the chains drawn by the pCN, the ApCN and the hybrid methods at $t=0.4$ and $t=0.8$.}
\label{f:acf_gauss1}
\end{figure}

\begin{figure}
\centerline{\includegraphics[width=.5\textwidth]{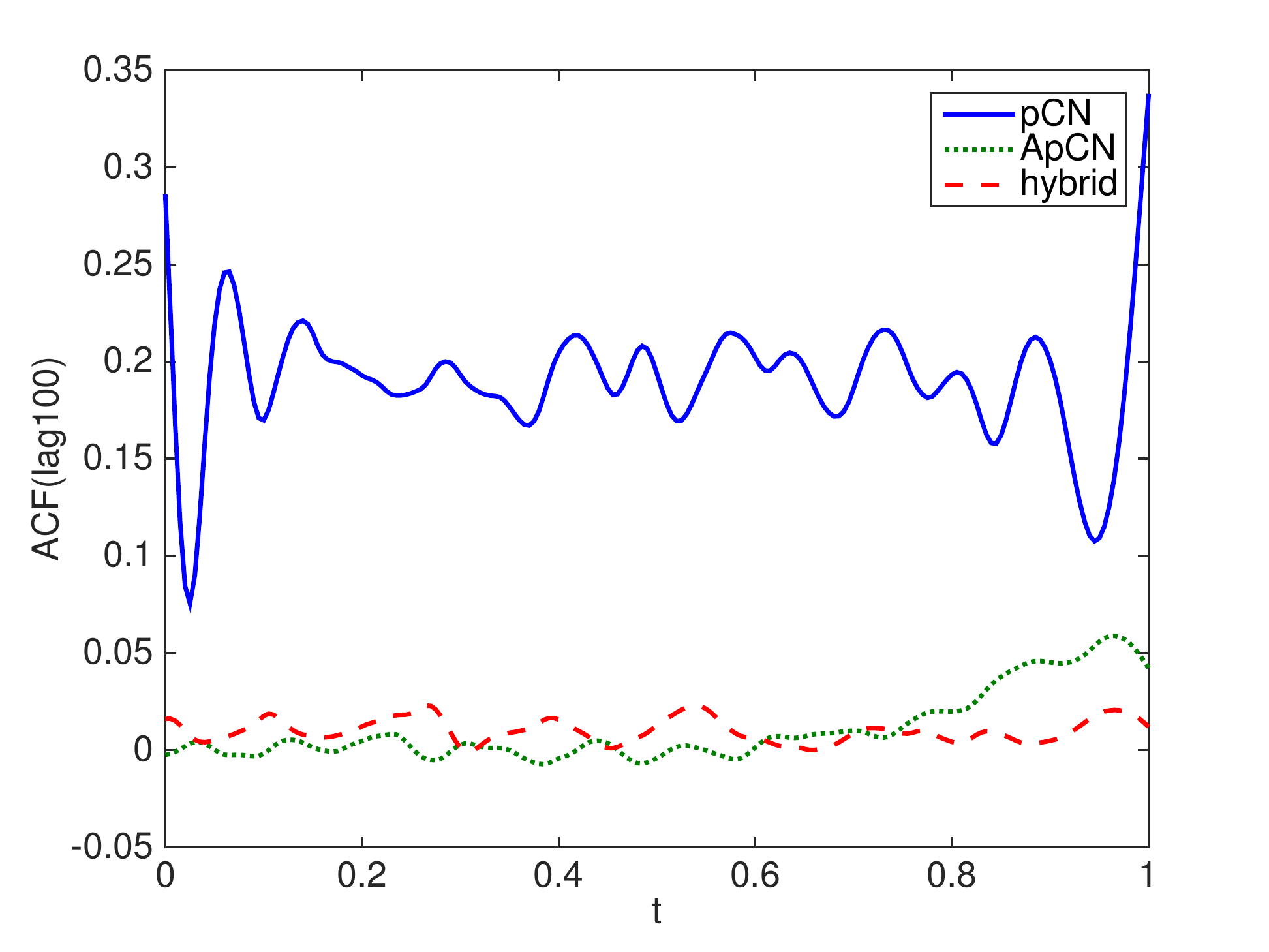}
\includegraphics[width=.5\textwidth]{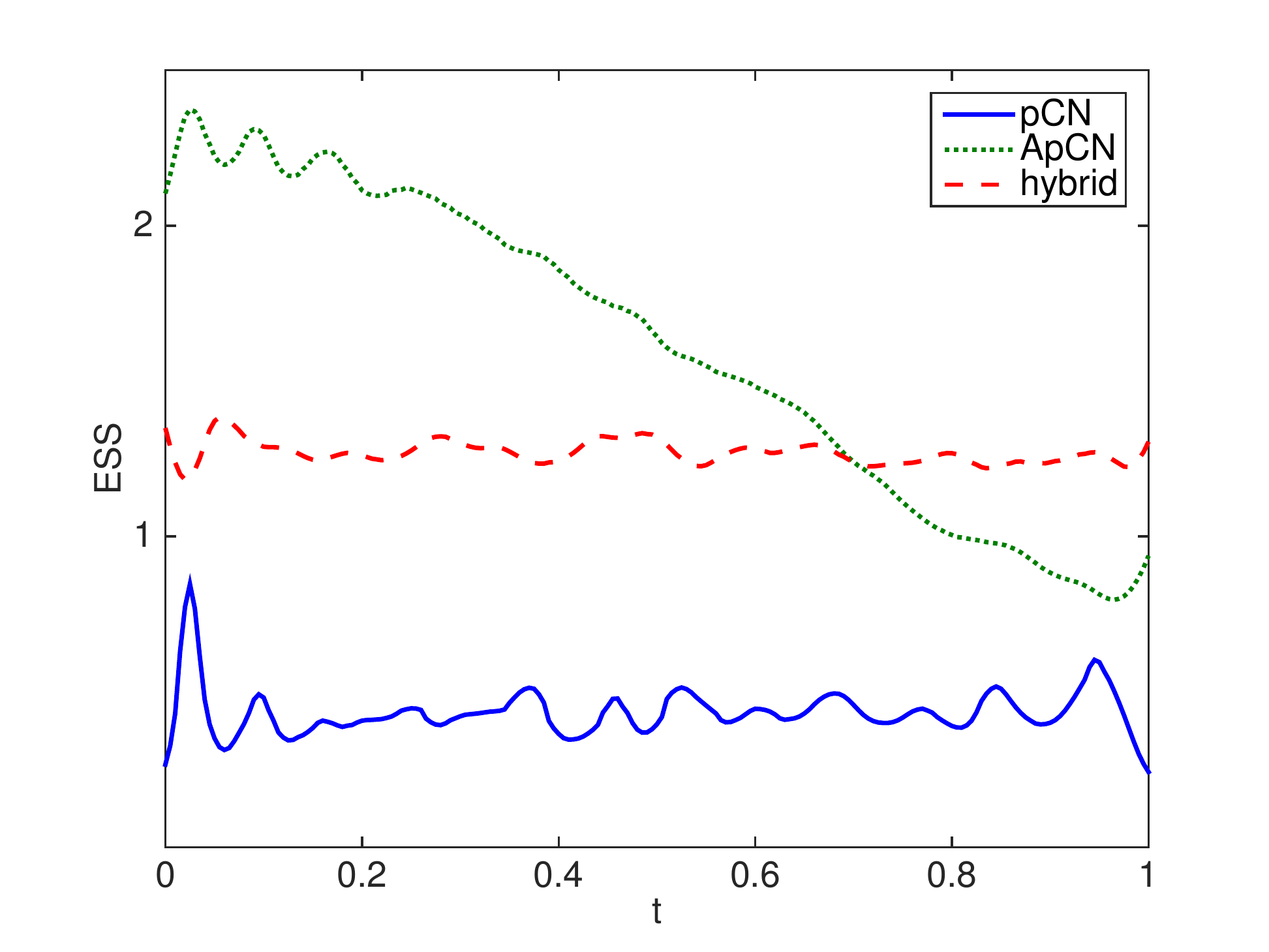}}
\caption{(for the Gaussian example: $\Delta=1$) Left: ACF (lag 100) at each grid point. Right: ESS per 100 samples at each grid point.}
\label{f:acf100-ess-gauss1}
\end{figure}

\begin{figure}
\centerline{\includegraphics[width=.5\textwidth]{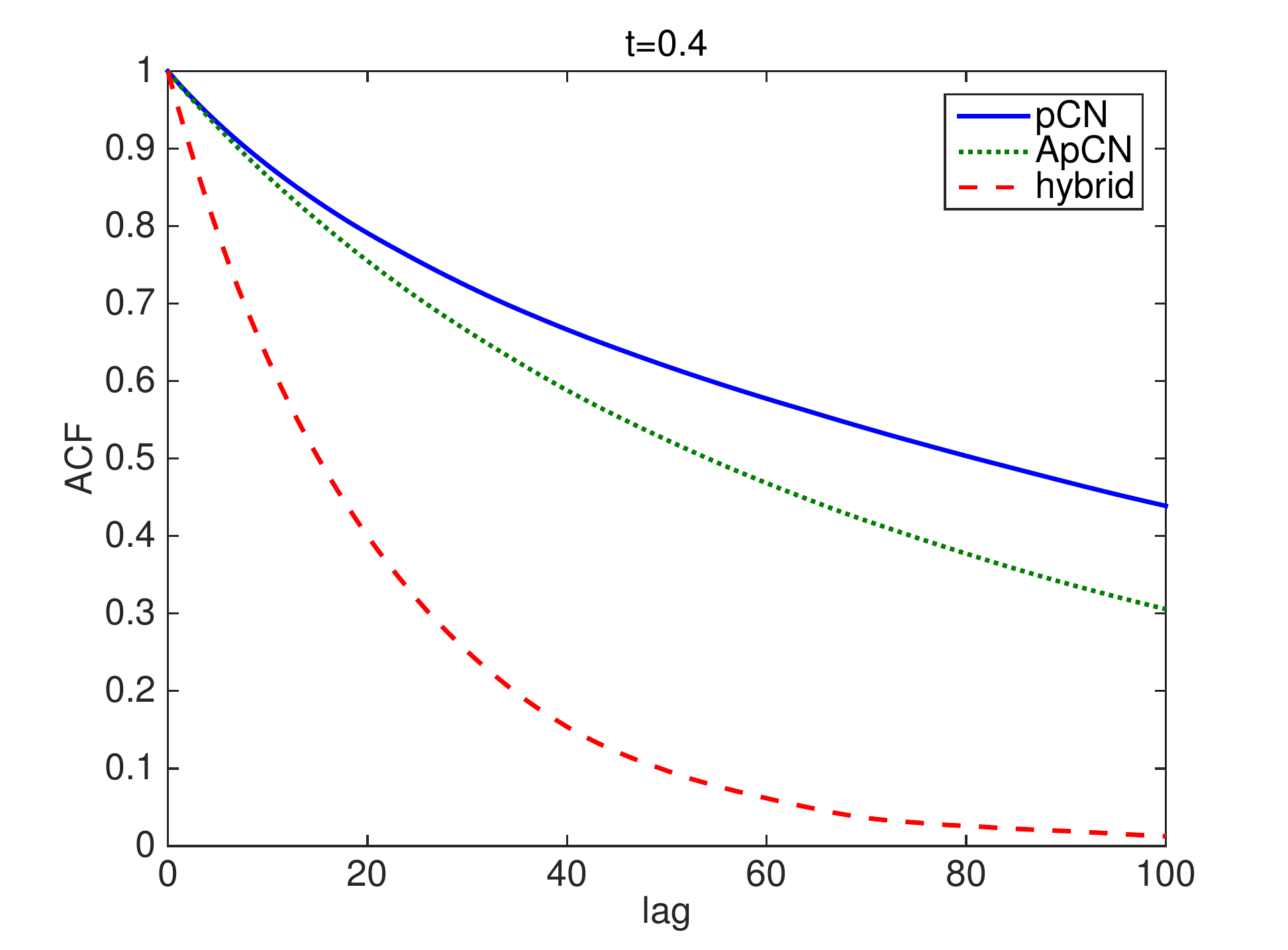}
\includegraphics[width=.5\textwidth]{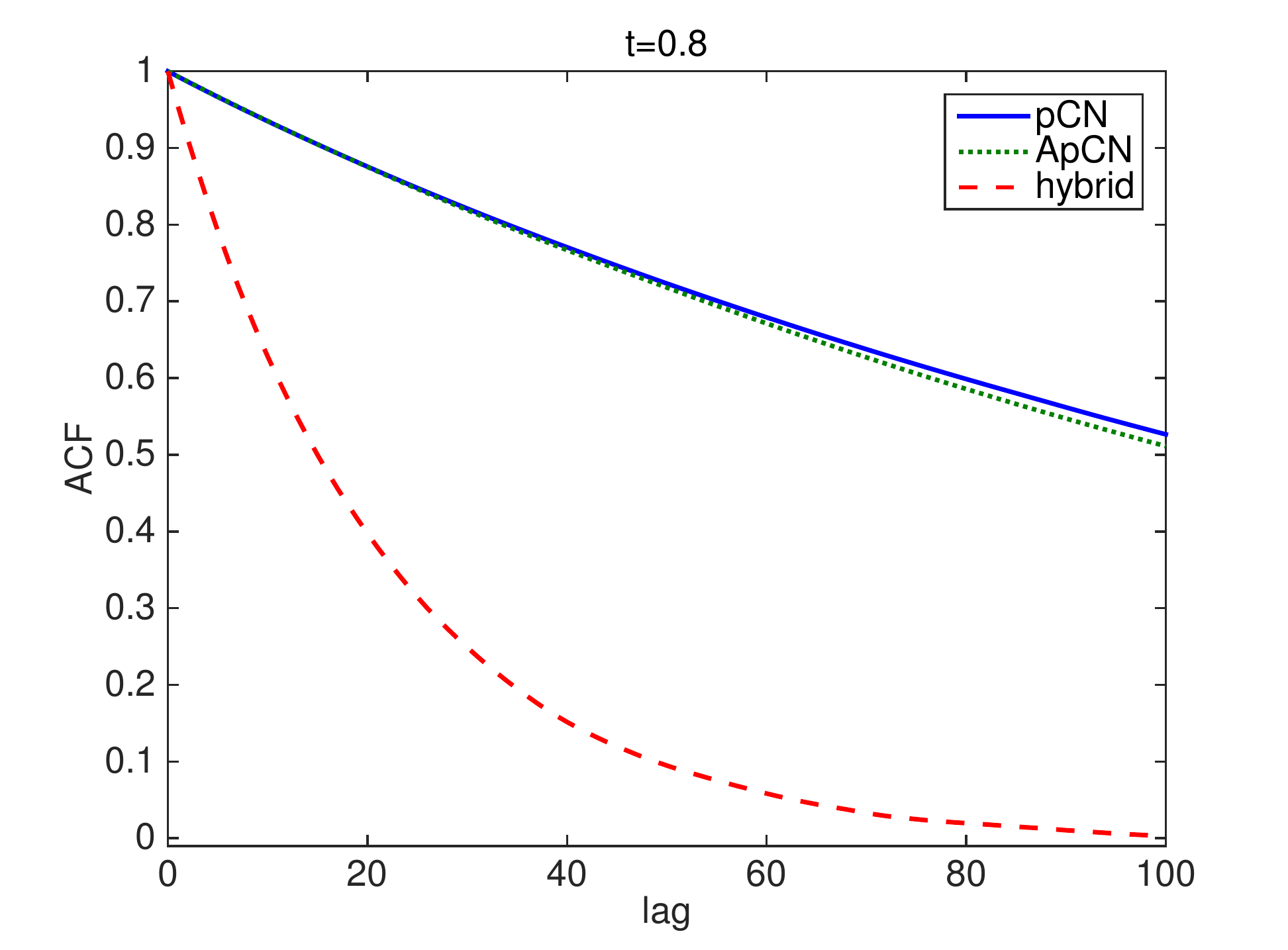}}
\caption{(for the Gaussian example: $\Delta=14$) ACF for the chains drawn by the pCN, the ApCN and the hybrid methods at $t=0.4$ and $t=0.8$.}
\label{f:acf_gauss14}
\end{figure}

\begin{figure}
\centerline{\includegraphics[width=.5\textwidth]{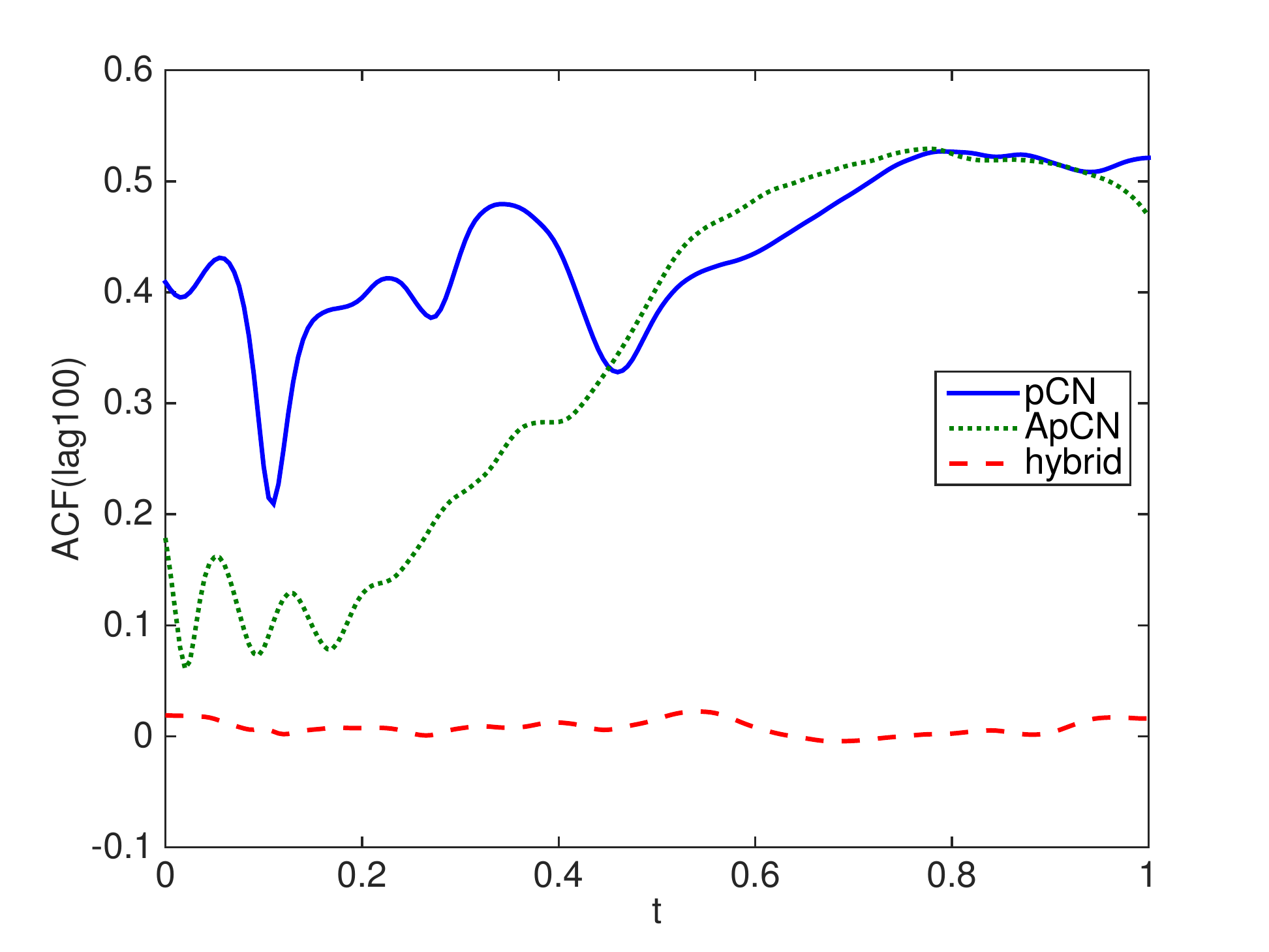}
\includegraphics[width=.5\textwidth]{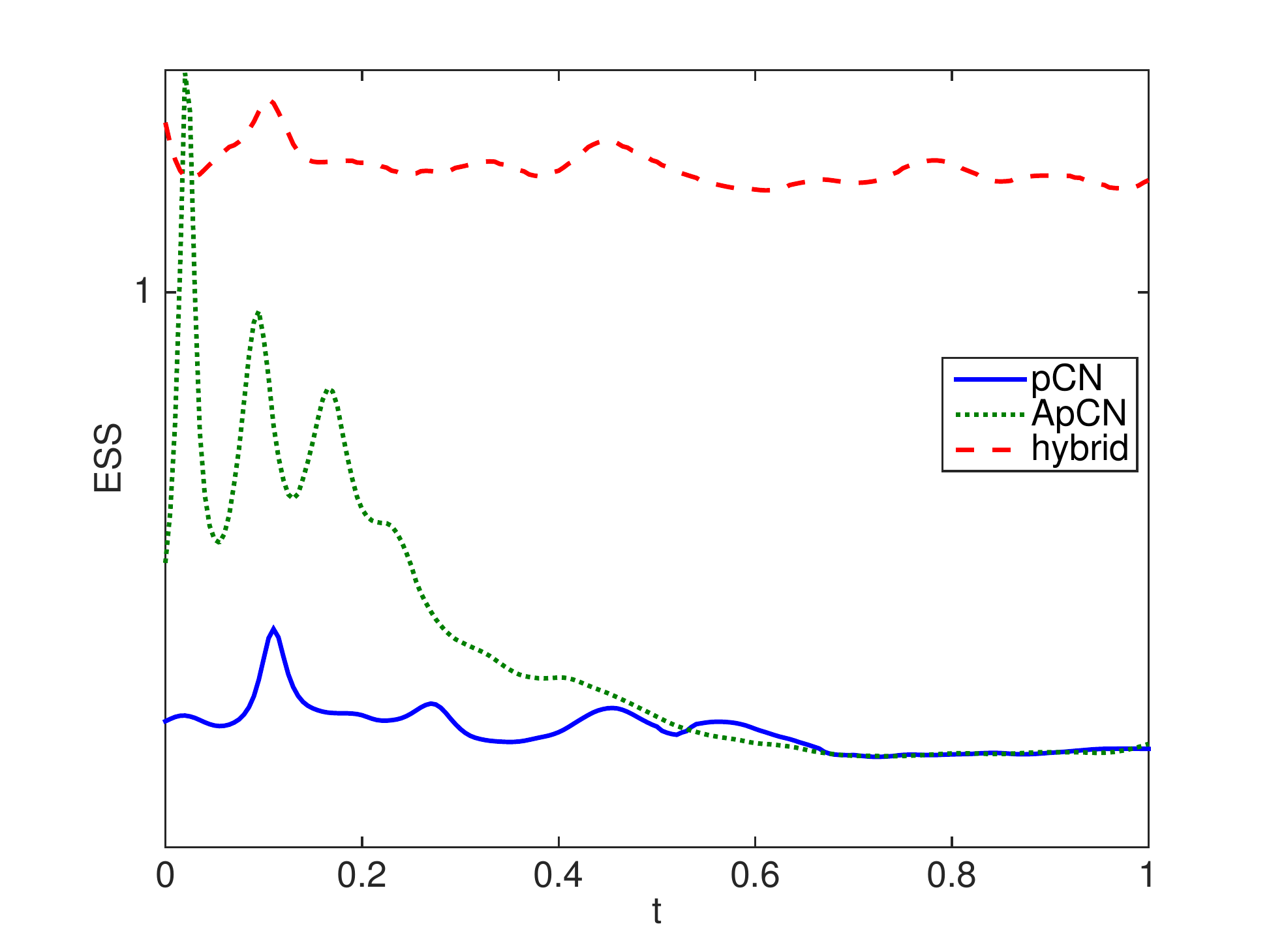}}
\caption{(for the Gaussian example: $\Delta=14$) Left: ACF (lag 100) at each grid point. Right: ESS per 100 samples at each grid point.}
\label{f:acf100-ess-gauss14}
\end{figure}

\subsection{An ODE example} Our second example is an inverse problem where the forward model is governed by an ordinary differential equation (ODE):
\[\frac{\partial x(t)}{\partial t} = -u(t)x(t) \]
with a prescribed initial condition. Suppose that we observe the solution $x(t)$ several times in the interval $[0,T]$, and we want to infer the unknown coefficient 
$u(t)$ from the observed data.
In our numerical experiments, we let the initial condition be $x(0) = 1$ and $T = 1$. Now suppose that the solution is measured every $T/50$ time unit from $0$ to $T$ and the error in each measurement is assumed to be an independent Gaussian $N(0,0.1^2)$. 
The prior is taken to be a zero mean Gaussian with covariance specified by Eq.~\eqref{e:matern}. 

First we want to compare the performance of the hybrid method with that of the ApCN method introduced in \cite{hu2015adaptive}, and so we use the same problem setup as is in \cite{hu2015adaptive}:
 we choose  $l=1$, $\sigma=1$.
We also use the same true coefficient $u(t)$ and
synthetic data $x(t)$ as those in \cite{hu2015adaptive},  which are shown in Figs.~\ref{f:data_ode}. 
We draw samples from the posterior with three methods: pCN, ApCN and the hybrid algorithm.
In both ApCN and the hybrid algorithm, we use $5\times10^5$ samples with additional $0.5\times10^5$ pCN samples used in the pre-run, and in the standard pCN we directly draw $5.5\times10^{5}$ samples.  
In both the ApCN and the hybrid methods, we follow \cite{hu2015adaptive}, 
and choose $J=14$, i.e., $14$ eigenvalues being adapted. 
Since the inference results have been reported in \cite{hu2015adaptive}, we omit them here and only compare
the performance of the three methods.
In Fig.~\ref{f:acf_ode}, we plot the ACF of the samples drawn by each method against the lag.
The results indicate that the ACF of both adaptive algorithms decay faster than the standard pCN, while the ACF of the hybrid algorithm 
decays faster than that of the ApCN. 
We then compute the ACF of lag $100$ at all the grid points, and show the results in Fig.~\ref{f:acf100-ess-ode} (left),  
and we can see that, the ACF of the chain generated by the hybrid method is clearly lower than that of the standard pCN and the ApCN at all the grid points. 
We compute the ESS per 100 samples of the unknown $u$ at each grid point and show the results in Fig.~\ref{f:acf100-ess-ode} (right).
The results show that the hybrid algorithm  produces much more effectively independent samples than  pCN and ApCN. 
In summary, with this example, we show that the proposed hybrid adaptive method performs better than both the standard pCN and the ApCN methods. 

\begin{figure}
\centerline{\includegraphics[width=.5\textwidth]{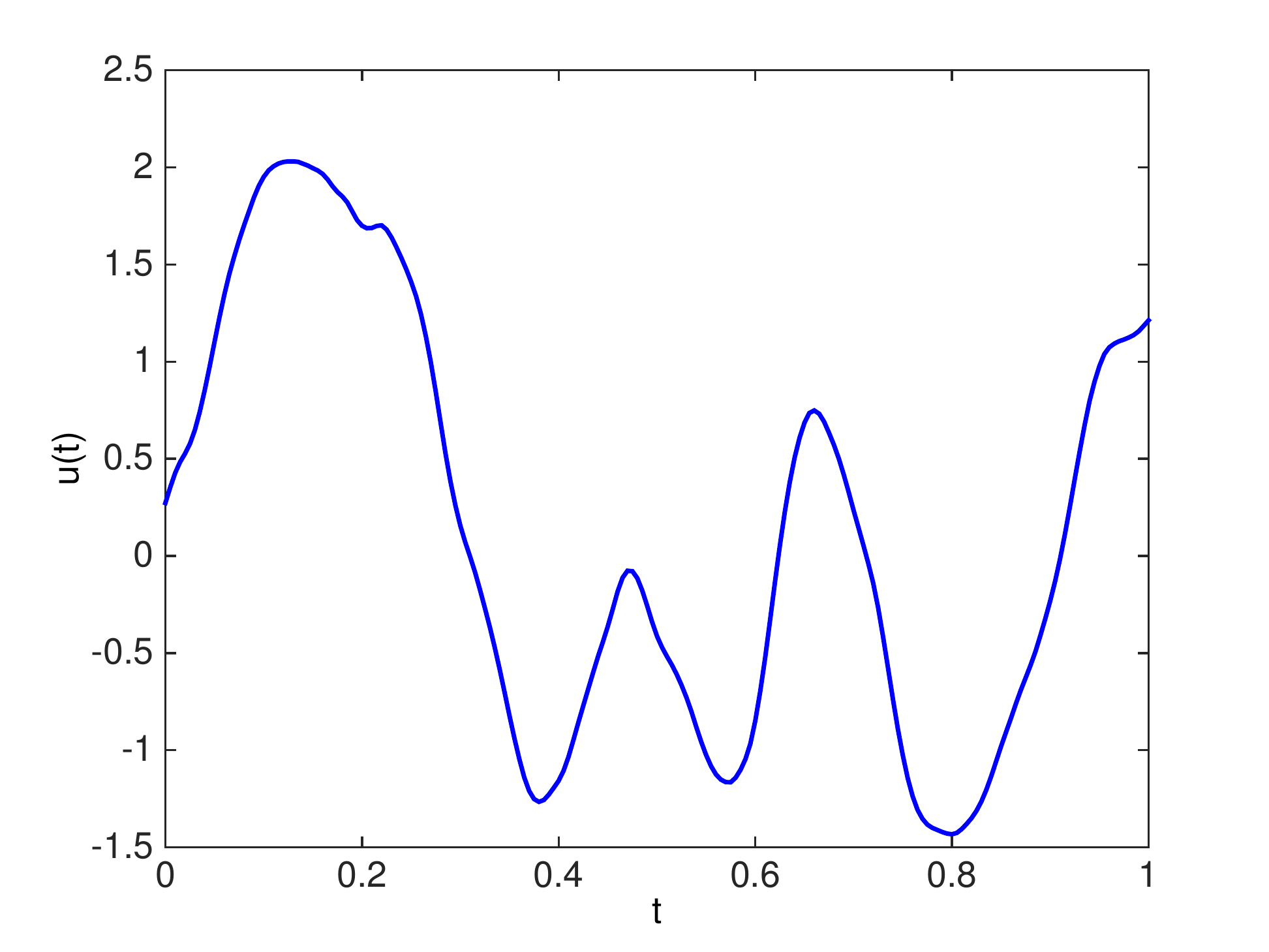}
\includegraphics[width=.5\textwidth]{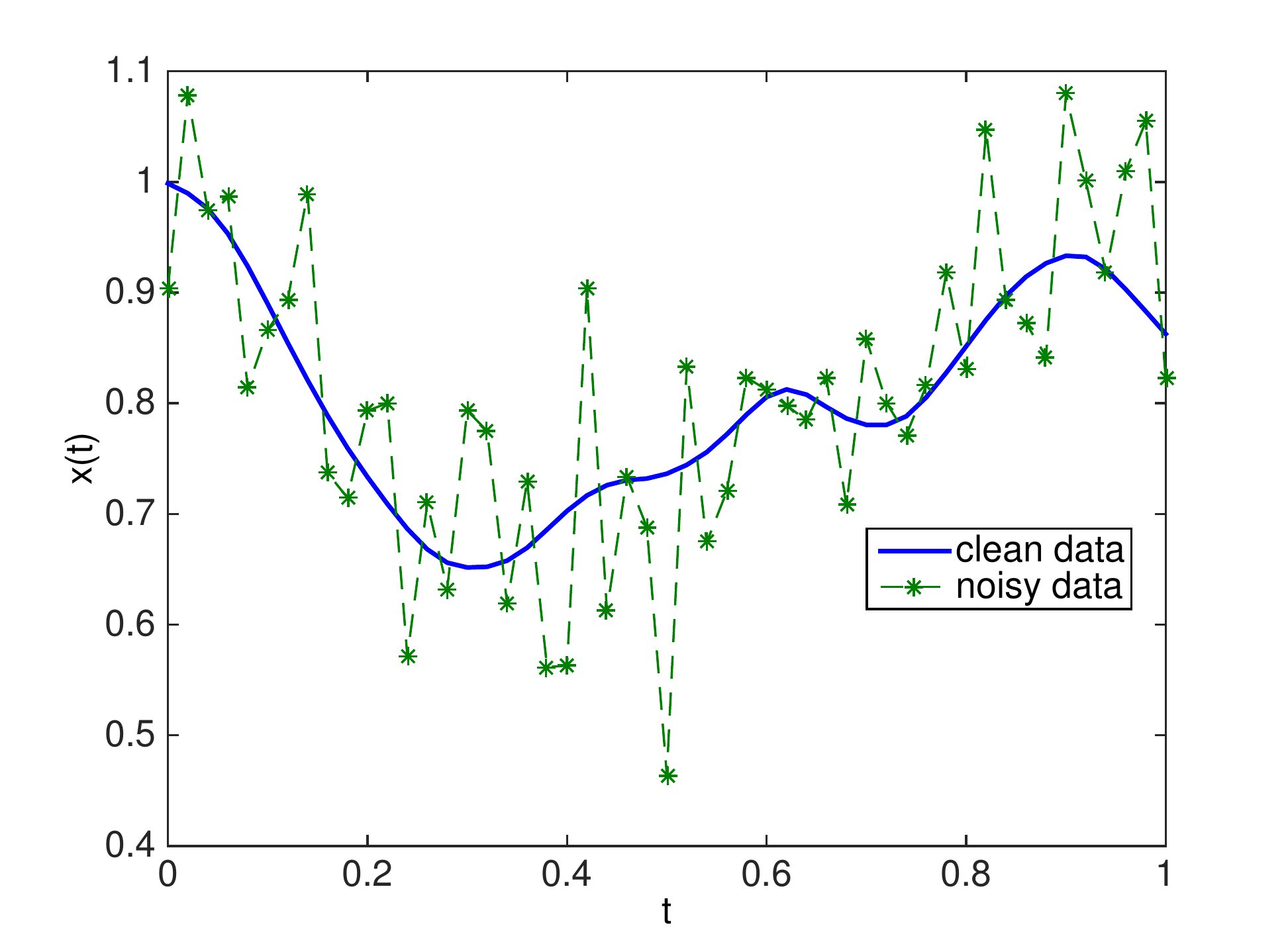}}
\caption{(for the ODE example: test 1) The truth (Left) and the data simulated with it (Right).}
\label{f:data_ode}
\end{figure}

\begin{figure}
\centerline{\includegraphics[width=.5\textwidth]{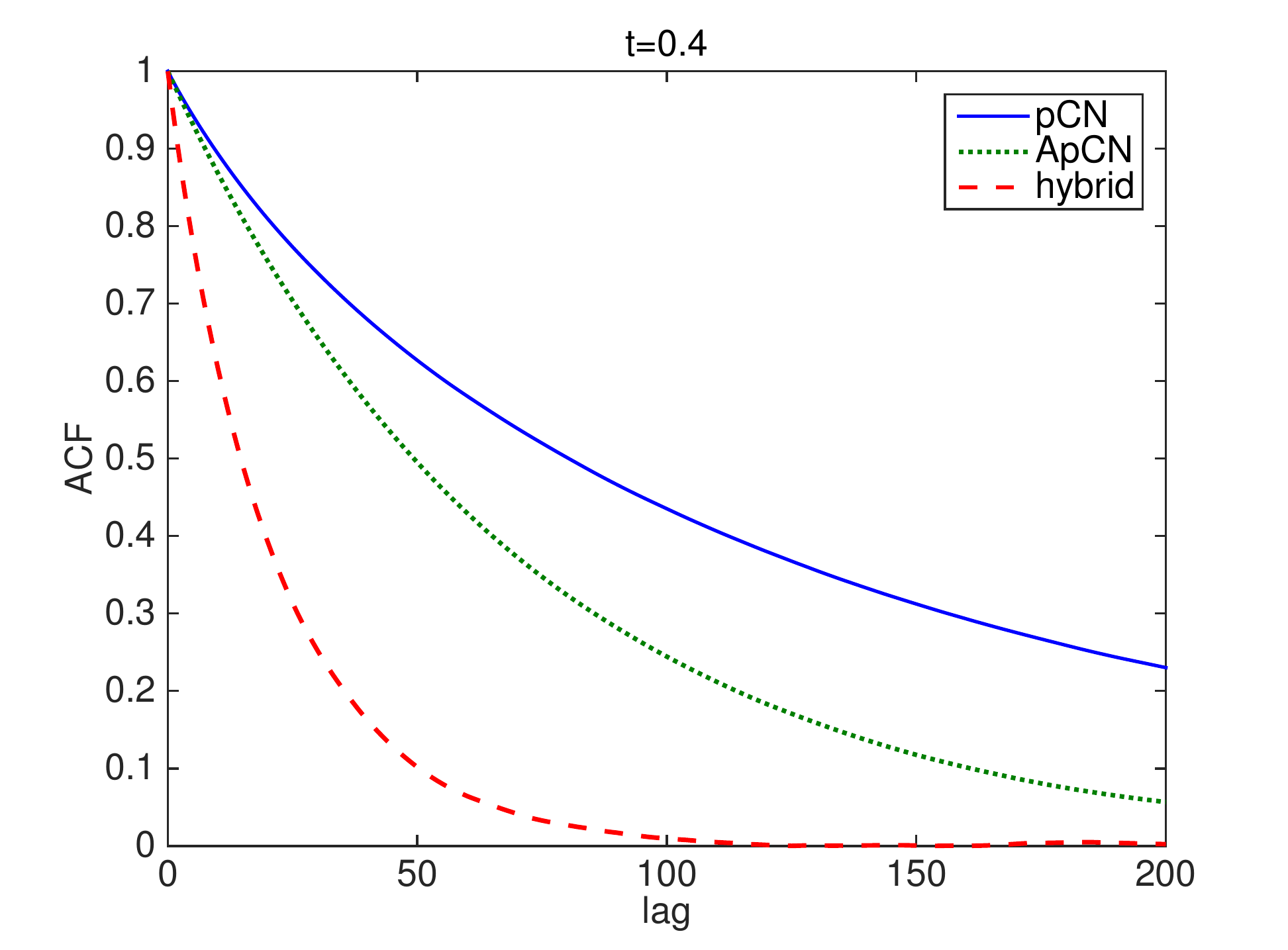}
\includegraphics[width=.5\textwidth]{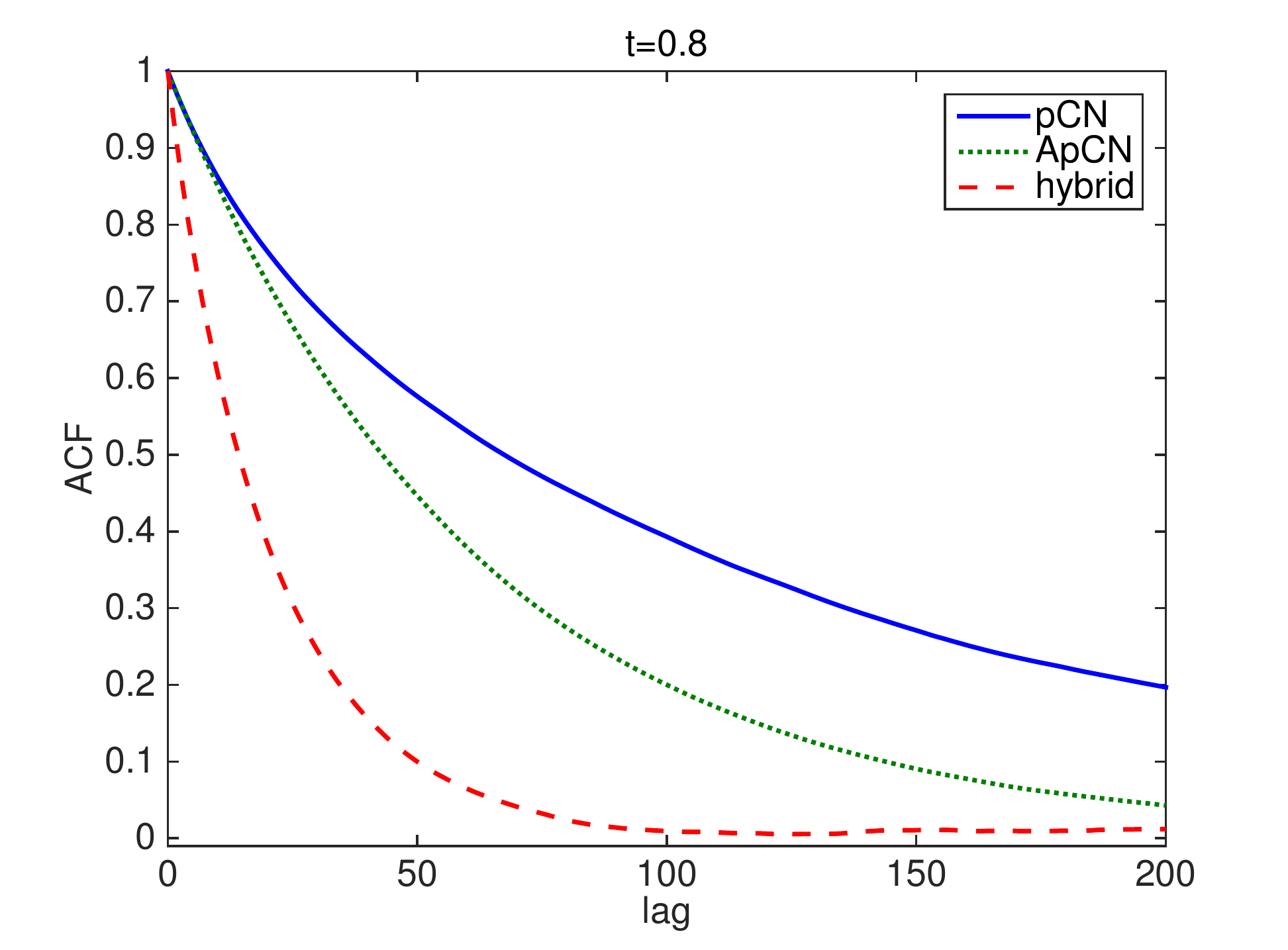}}
\caption{(for the ODE example: test 1) ACF for the chains drawn by the pCN, the ApCN and the hybrid methods at $t=0.4$ and $t=0.8$.}
\label{f:acf_ode}
\end{figure}

\begin{figure}
\centerline{\includegraphics[width=.5\textwidth]{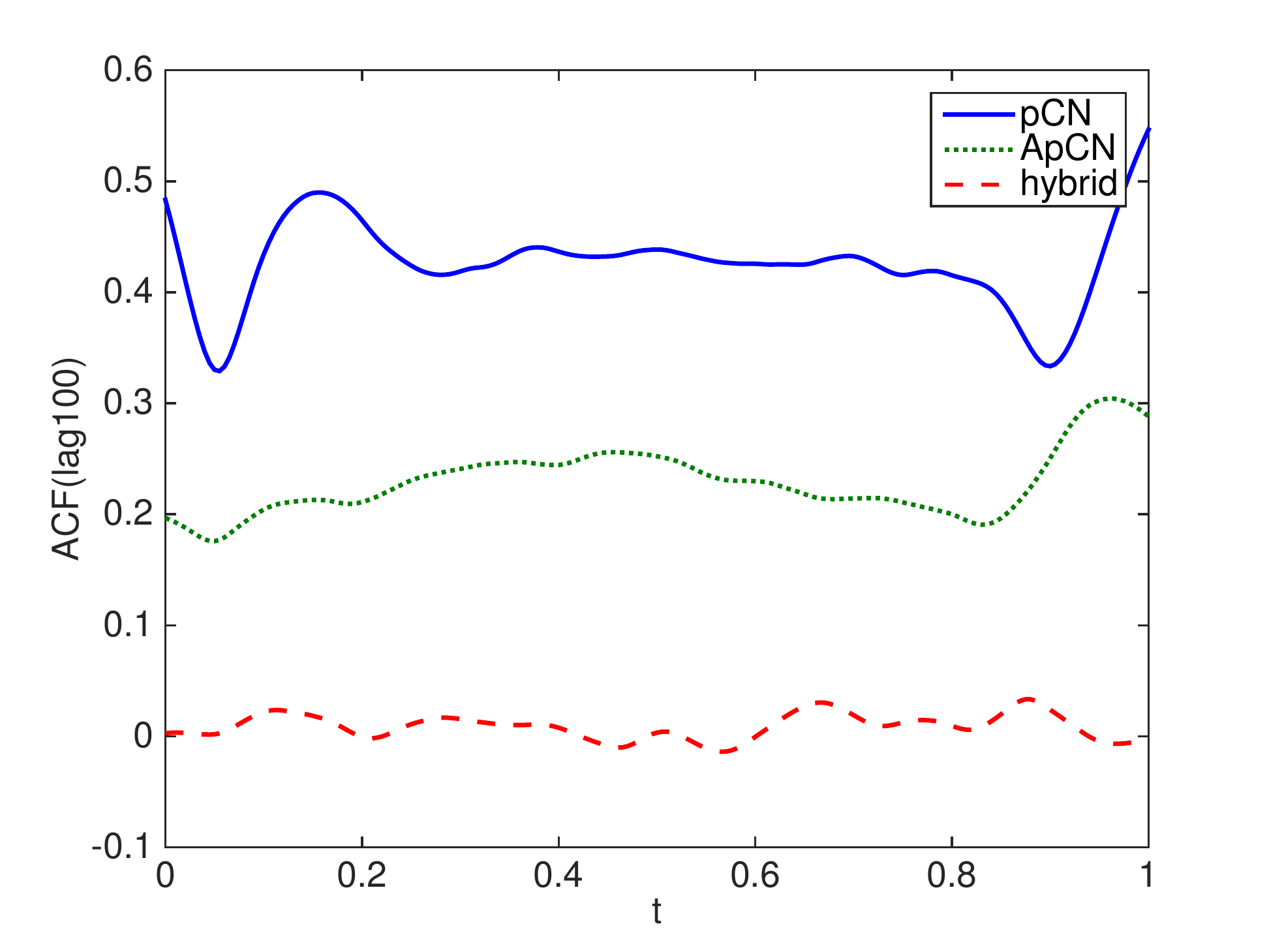}
\includegraphics[width=.5\textwidth]{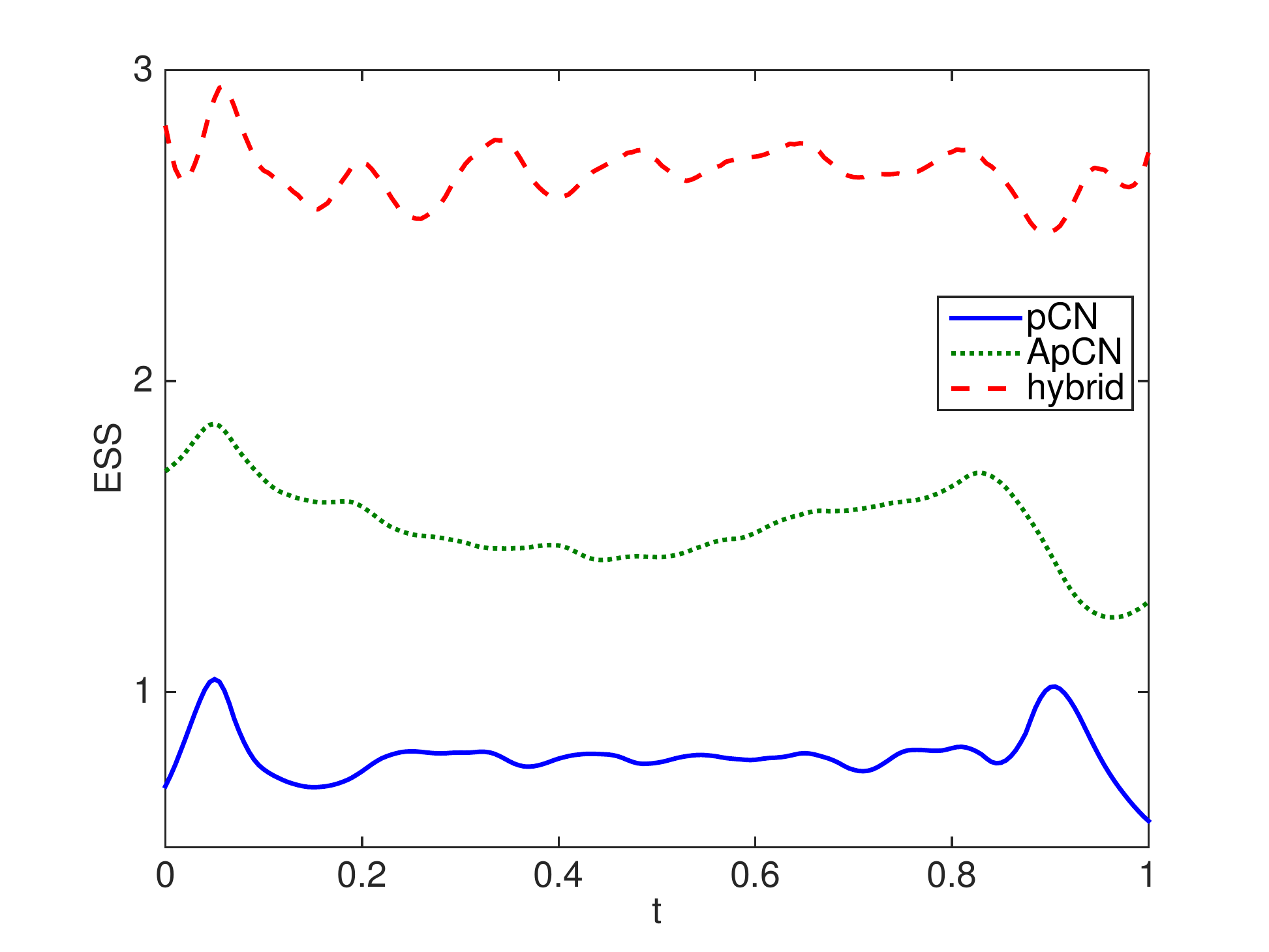}}
\caption{(for the ODE example: test 1) Left: ACF (lag 100) at each grid point. Right: ESS per 100 samples at each grid point.}
\label{f:acf100-ess-ode}
\end{figure}

Next we use the example to test how the value of $J$ affects the sampling efficiency. 
We choose  $l=0.2$, $\sigma=1$. A ``truth'' $u(t)$ is randomly generated from the prior distribution and the
synthetic data $x(t)$ is
 generated by applying the forward model to the generated coefficient $u$ and then adding noise to the result.  
Both the simulated data and the truth are shown in Fig.~\ref{f:data_J}.
It can be seen that the true coefficient and the prior are much ``rougher'' than those in the previous test.
We perform the hybrid algorithm with three different values of $J$: $J=5$, $J=10$ and $J=20$, each with
$5\times10^{5}$  plus $0.5\times10^5$ pCN (pre-run) samples.
As a comparison, we also perform a standard pCN with $5.5\times10^4$ samples. 
The posterior mean is shown in Fig.~\ref{f:data_J} (right). 
We plot the ACF as a function of lag at $t=0.4$ and $t=0.8$ for all the results in Figs~\ref{f:acf_J}.
In Figs.~\ref{f:acf100_J} we plot the ACF of lag 100 as well as the ESS at all the grid points.
One can see from the plots that, the algorithm with $J=10$ yields the best results, suggesting 
that $J=10$ may be sufficient for this problem and $J=20$ may be too large for the given number of samples. 
Nevertheless, in all the cases, the hybrid method performs better than the standard pCN.

\begin{figure}
\centerline{
\includegraphics[width=.5\textwidth]{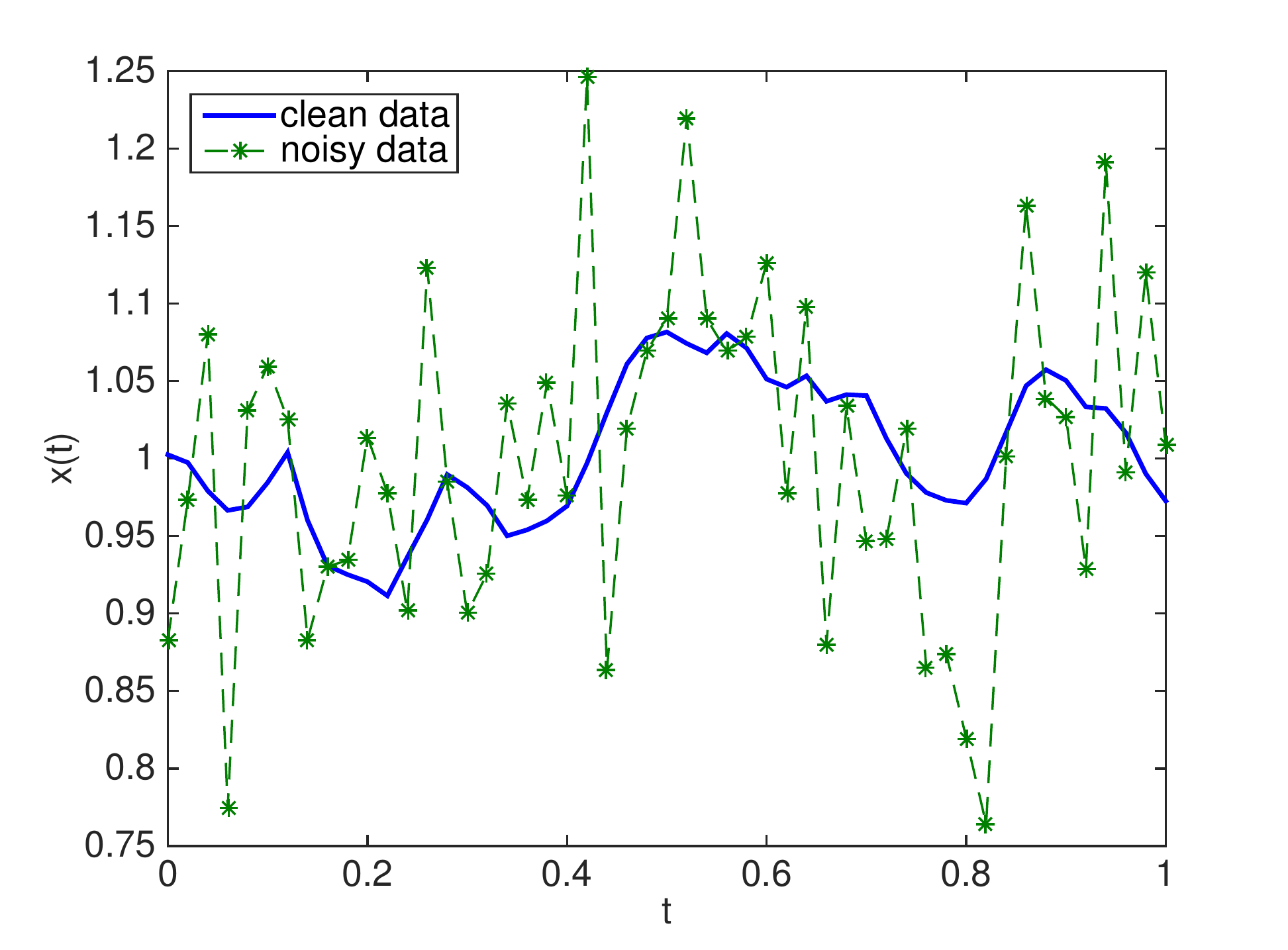}
\includegraphics[width=.5\textwidth]{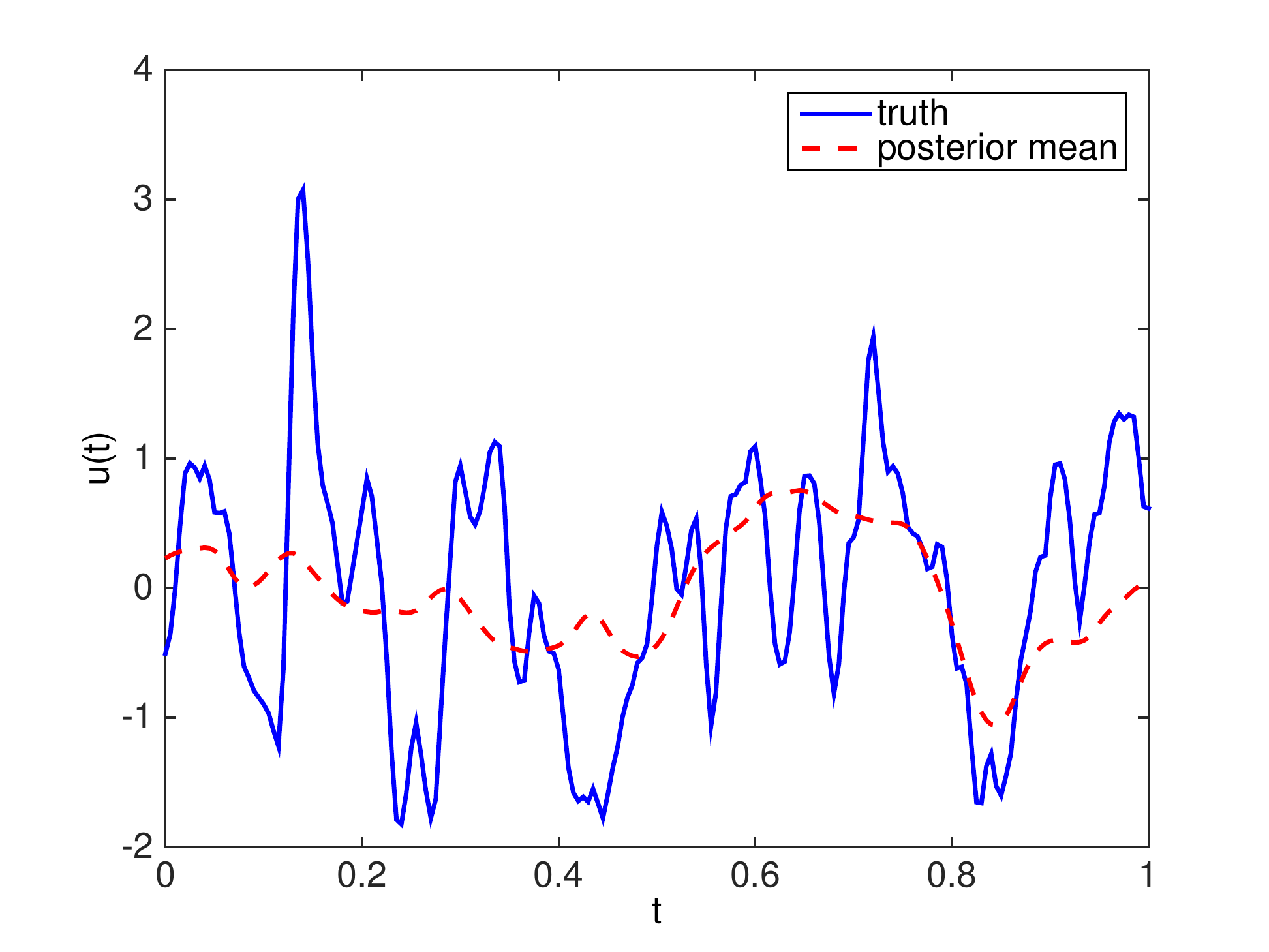}}
\caption{(for the ODE example: test 2) Left: the simulated data. Right: the posterior mean (dashed) compared to the truth (solid).}
\label{f:data_J}
\end{figure}


\begin{figure}
\centerline{\includegraphics[width=.5\textwidth]{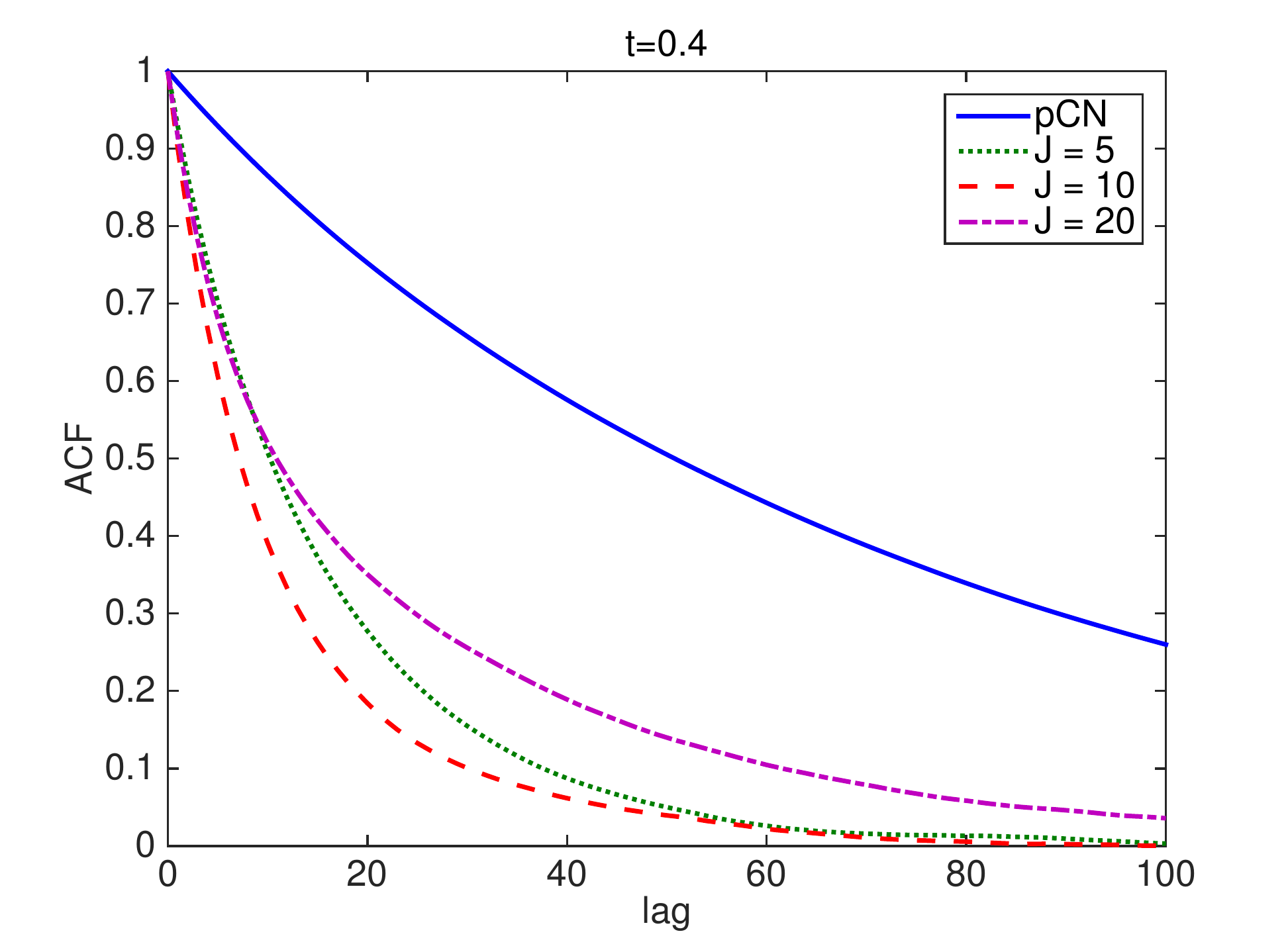}
\includegraphics[width=.5\textwidth]{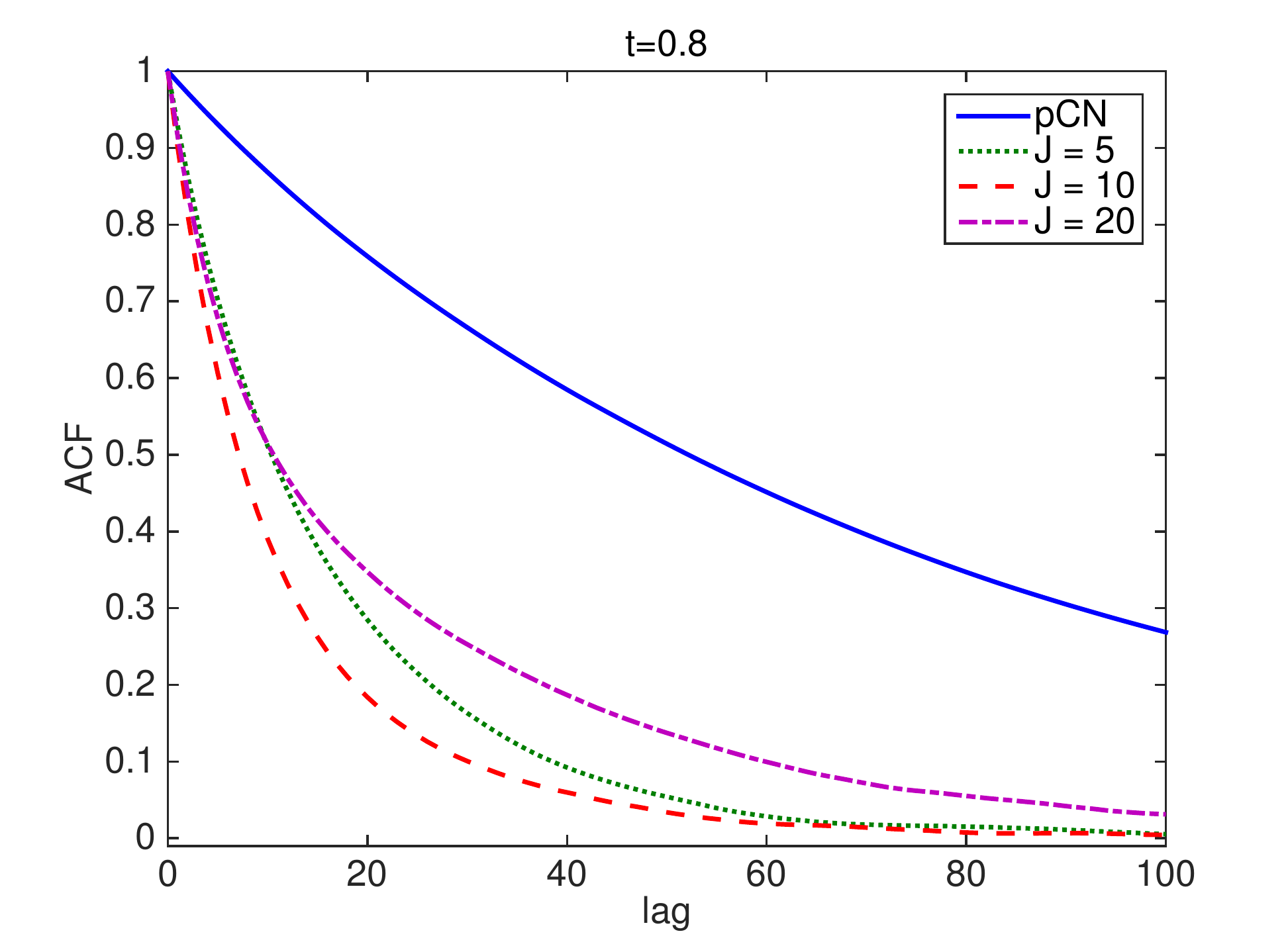}}
\caption{(for the ODE example: test 2) ACF for the chains drawn by the hybrid method with $J=5,\,10,\,20$ at $t=0.4$ and $t=0.8$.}
\label{f:acf_J}
\end{figure}

\begin{figure}
\centerline{\includegraphics[width=.5\textwidth]{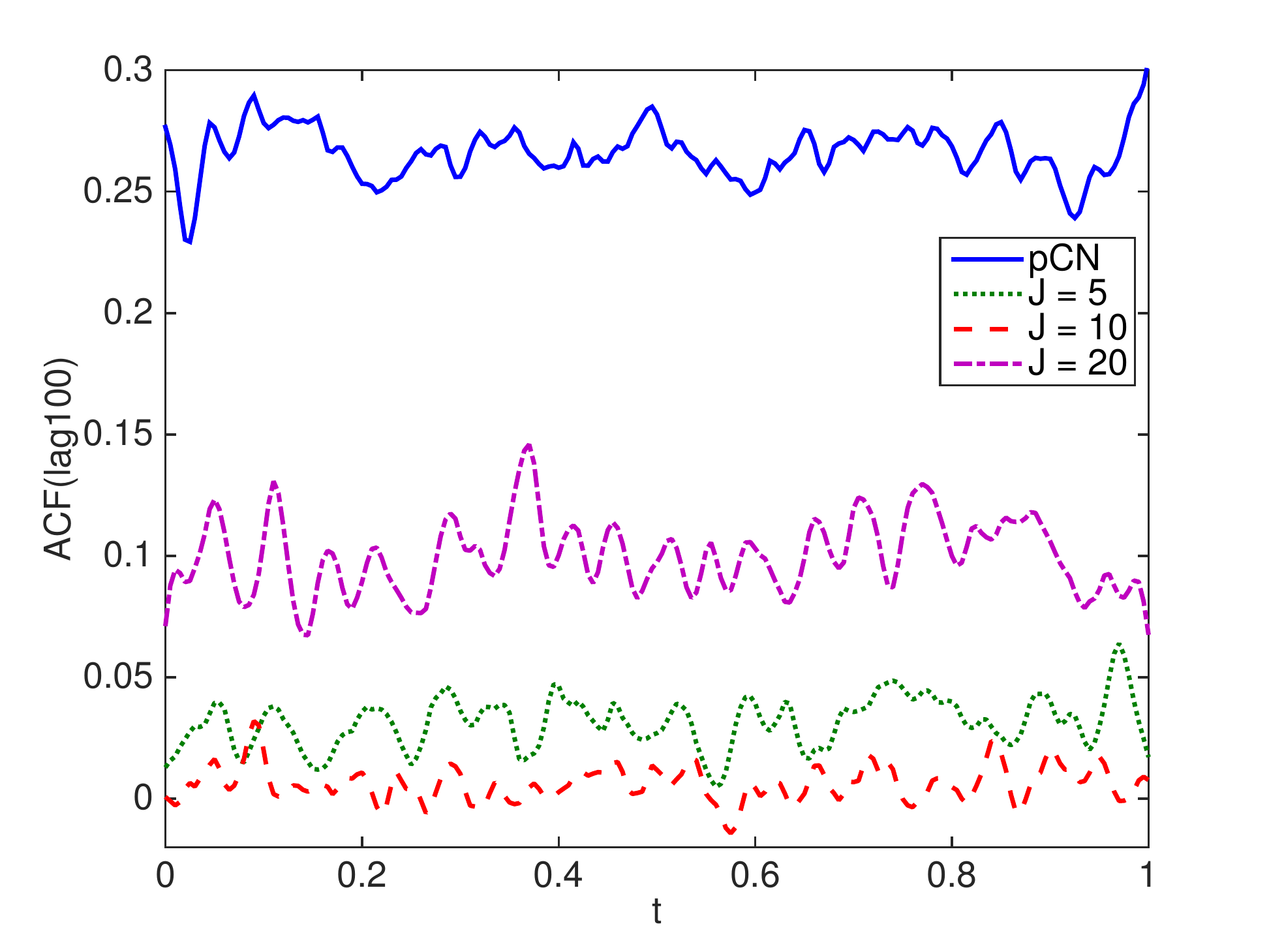}
\includegraphics[width=.5\textwidth]{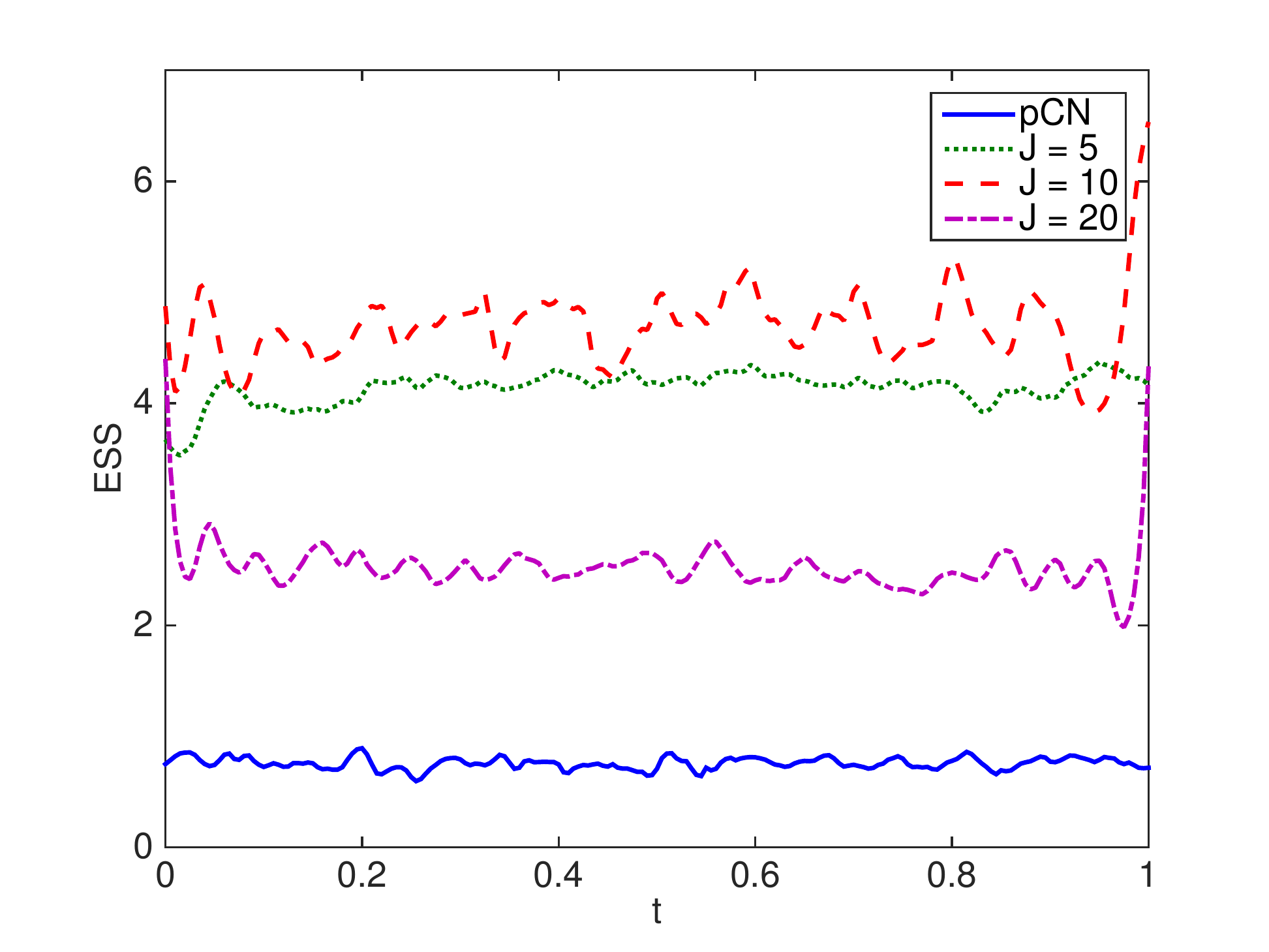}}
\caption{(for the ODE example: test 2) Left: ACF (lag 100) at each grid point for $J=5,\,10,\,20$. Right: ESS per 100 samples at each grid point for $J=5,\,10,\,20$.}
\label{f:acf100_J}
\end{figure}


\subsection{Estimating the Robin coefficient}
In the last example, we consider a one-dimensional heat conduction equation in the region $x\in [0,L]$ ,
\begin{subequations}
\label{e:heat}
\begin{align}
&\frac{\partial u}{\partial t}(x,t) = \frac{\partial^2 u}{\partial x^2}(x,t), \\
&u(x,0)=g(x), 
\end{align}
with the following Robin boundary conditions:
\begin{align}
&-\frac{\partial u}{\partial x}(0,t) + \rho(t) u(0,t) = h_0(t),\\
 &\frac{\partial u}{\partial x}(L,t) + \rho(t) u(L,t) = h_1(t).
\end{align}
\end{subequations}
Suppose the functions $g(x)$,  $h_0(x)$ and $h_1(x)$ are all known, and we want to estimate the unknown Robin coefficient $\rho(t)$
from certain measurements of the temperature $u(x,t)$.  This example is studied in \cite{yao2016tv,hu2015adaptive,yang2009identification}.
Here we choose $L=1$, $T=1$ and the functions to be
\[g(x)=x^2+1,\quad  h_0 =t(2t+1),\quad h_1=2+t(2t+2).\]
A temperature sensor is place at the end $x=0$. 
 The solution is measured every $T/200$ time unit from $0$ to $T$ and the error in each measurement is assumed to be an independent Gaussian $N(0,0.1^2)$.
 Moreover, the prior is the same as that used in the first test of the ODE example.

The data is generated the same as the ODE example, with the true Robin coefficient randomly drawn from the prior distribution. 
Both the truth and the simulated data are shown in Fig.~\ref{f:data_pde}.
We sample the posterior distribution with the three methods: pCN, ApCN and the hybrid algorithm.
In the ApCN and the hybrid methods, ocne again we choose $J=14$ and draw $5\times10^5$ (adaptive) + $0.5\times10^5$ (prerun) samples.
In the pCN method, we draw $5.5\times10^5$ samples directly. 
We show the obtained posterior mean in Fig~\ref{f:data_pde} (right). 
We now compare the performance of the three methods.
First we plot the ACF of the samples obtained by the methods at $t=0.1$ and $t=0.5$ in Fig~\ref{f:acf_pde},
and then we plot the ACF at lag 100 and the ESS at all the grid points in Fig~\ref{f:acf100-ess-pde}. 
In all these figures, we can see that, while both adaptive algorithms yield better results than the standard pCN, 
the hybrid algorithm clearly outperforms the ApCN method,
which again indicates that the new algorithm can significant improve the sampling efficiency over the ApCN approach, by taking the correlations 
between eigenmodes into account.

\begin{figure}
\centerline{
\includegraphics[width=.5\textwidth]{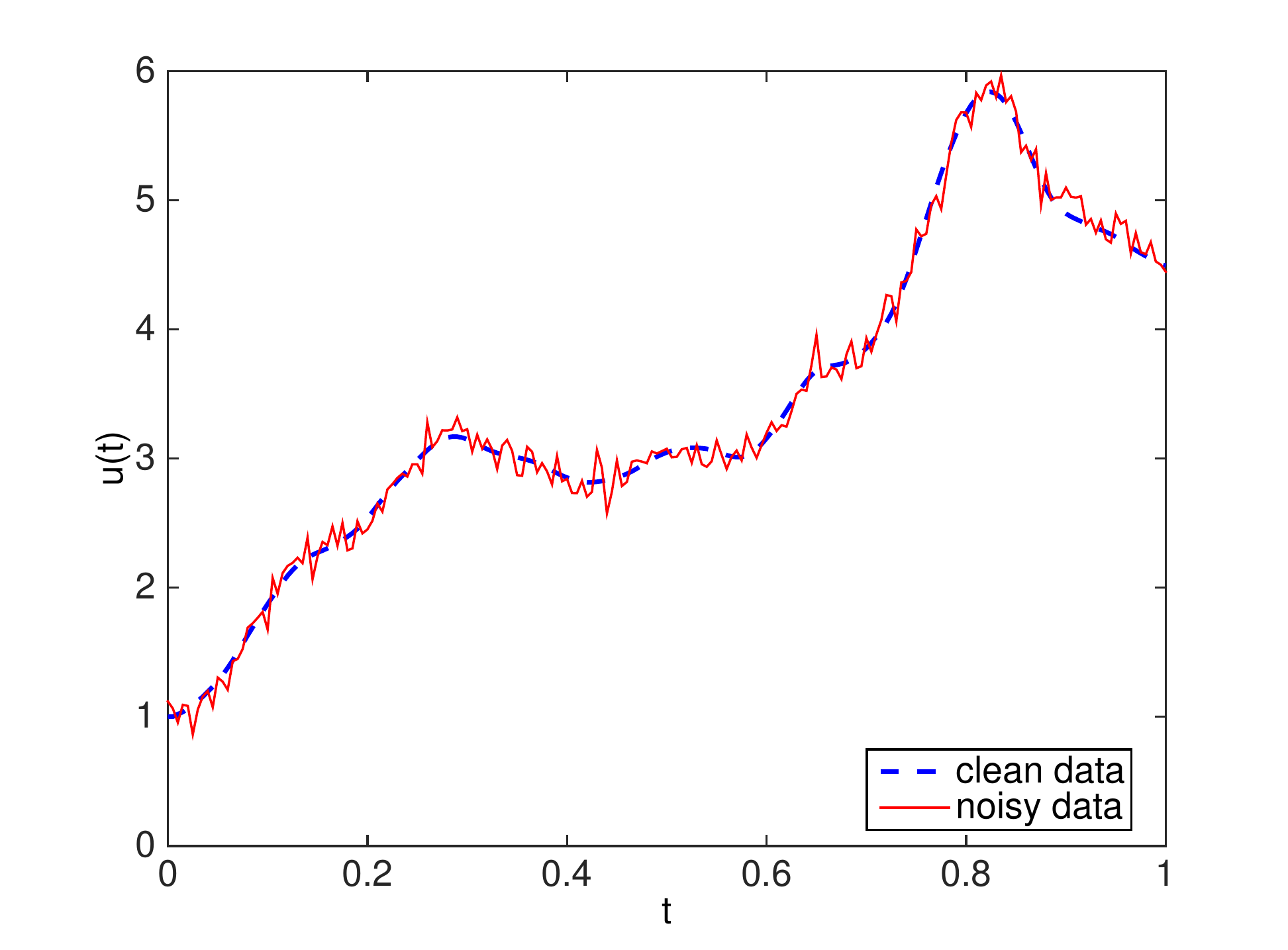}
\includegraphics[width=.5\textwidth]{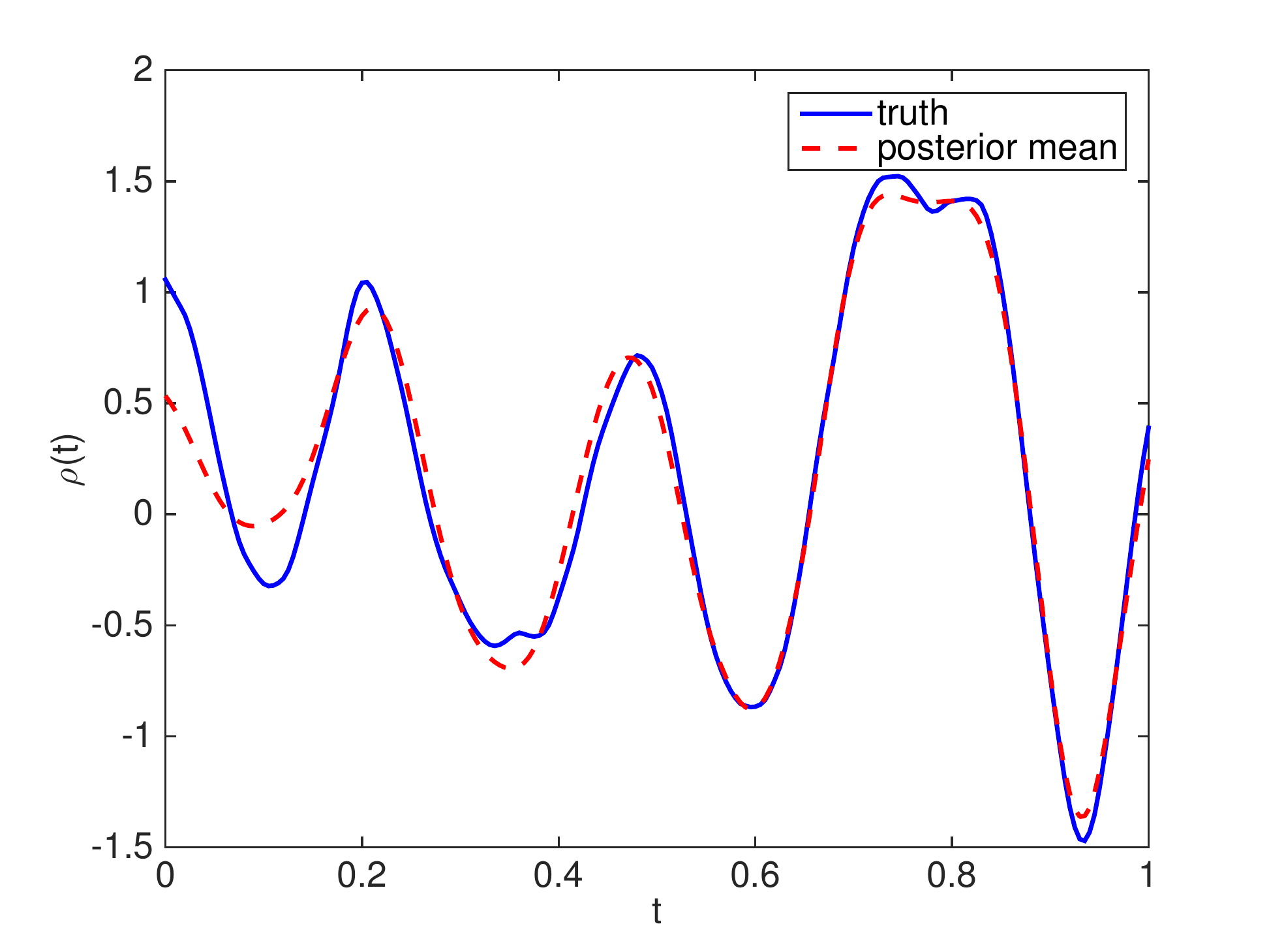}}
\caption{(for the Robin example) Left: the simulated data. Right: the truth and the posterior mean.}
\label{f:data_pde}
\end{figure}


\begin{figure}
\centerline{\includegraphics[width=.5\textwidth]{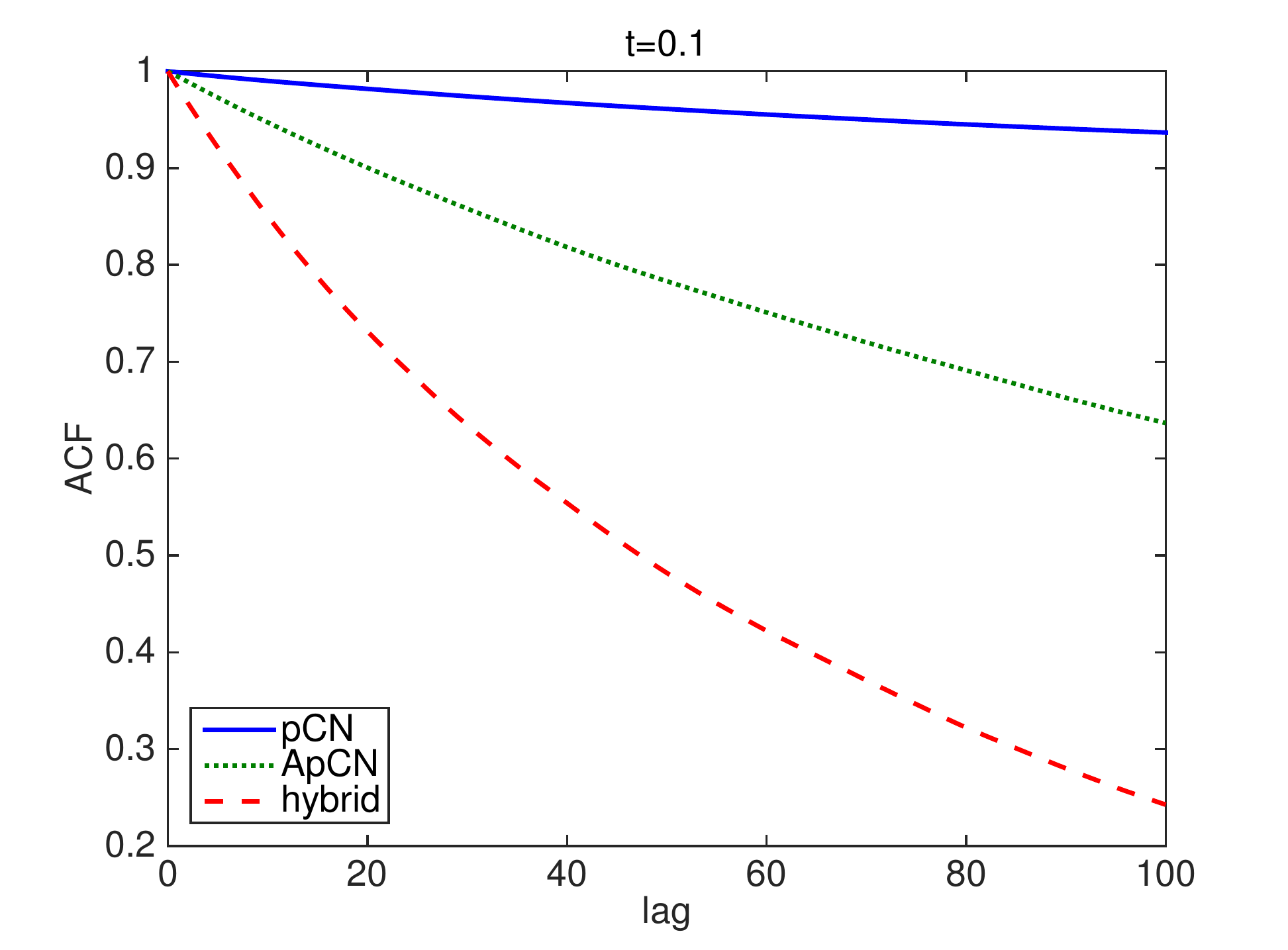}
\includegraphics[width=.5\textwidth]{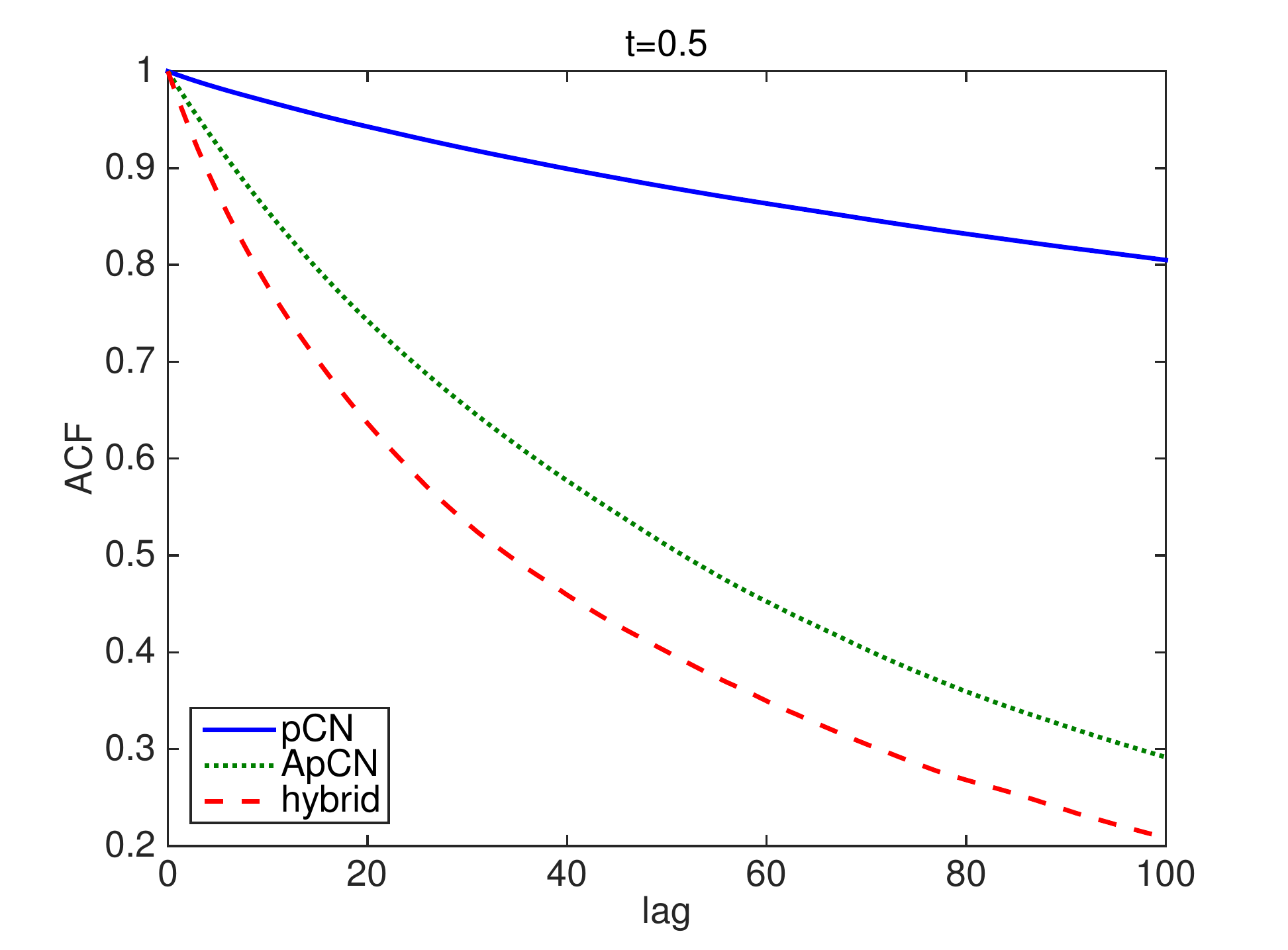}}
\caption{(for the Robin example) ACF for the pCN and the ApCN methods at $t=0.1$ (left) and $t=0.5$ (right).}
\label{f:acf_pde}
\end{figure}

\begin{figure}
\centerline{\includegraphics[width=.5\textwidth]{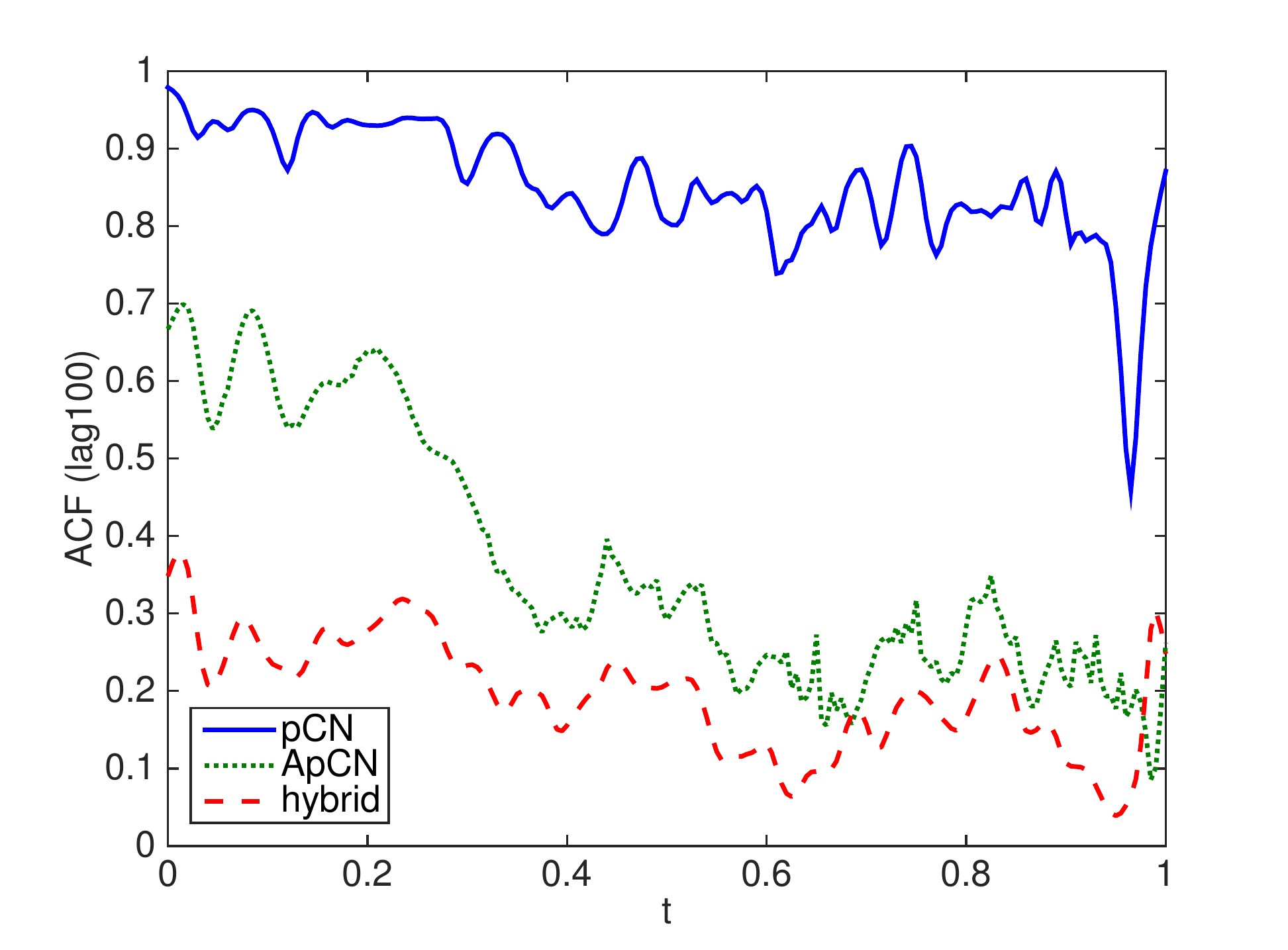}
\includegraphics[width=.5\textwidth]{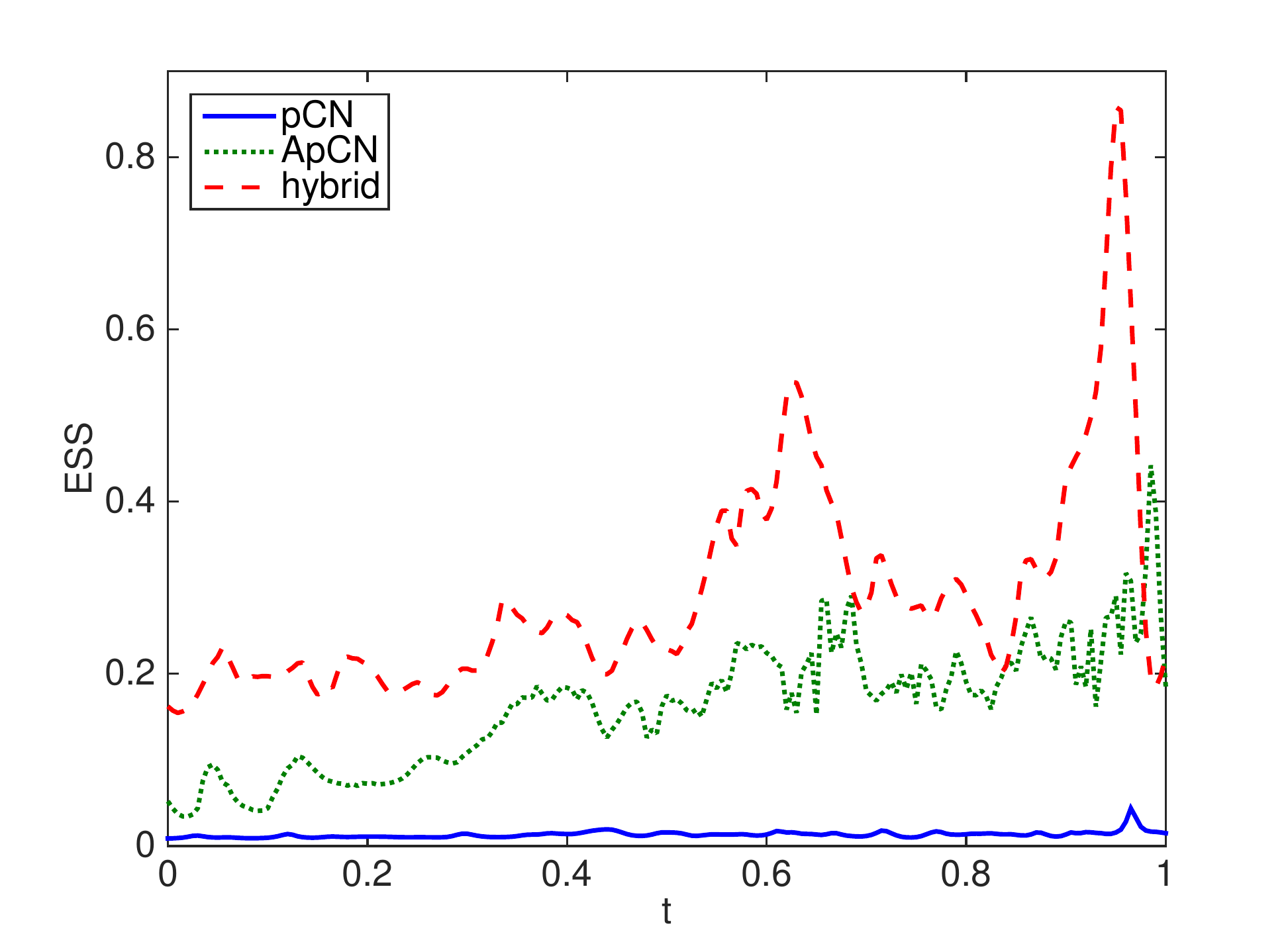}}
\caption{(for the Robin example) Left: ACF (lag 100) at each grid point. Right: the ESS at each grid point. }
\label{f:acf100-ess-pde}
\end{figure}

\section{Conclusions}\label{sec:conclusion}
In summary, we consider MCMC simulations for Bayesian inferences in function spaces. 
We develop a hybrid algorithm, which combines the adaptive Metropolis and the pCN algorithm, particularly addressing some limitations of our previously developed ApCN method. 
The implementation of the proposed algorithm is rather simple, without requiring any information of the underlying models.
We also show that the hybrid adaptive algorithm satisfies certain important ergodicity conditions without making any modifications of the likelihood function.
Finally we demonstrate the  efficiency of the hybrid adaptive algorithm with several numerical examples,
in which we see that the hybrid algorithm can evidently outperform the ApCN method, thanks to its ability to take into account the correlations between eigenfunctions. 
Note here that, in problems where the correlations between eigenfunctions are weak, the hybrid may not improve the efficiency much over the ApCN method.
Nevertheless, as is illustrated by our numerical examples, in that case, the hybrid algorithm's performance is at least comparable to that of the ApCN.  
We expect the hybrid adaptive algorithm can be useful in many applied problems, especially in those with underlying models whose gradient information is difficult to obtain.

Some improvements of the hybrid algorithm are possible. 
First, in the present formulation of the hybrid algorithm, we choose to adapt in the subspace spanned by the eigenfunctions corresponding to the leading eigenvalues.
This strategy can be improved by allowing the algorithm to automatically identify this ``data-informed subspace'' during the iterations. 
Moreover, reduced models or surrogates (see e.g. \cite{cui2015data,li2014adaptive,yan2015stochastic}) of the forward operator may be constructed and used in the subspace to accelerate the simulation. 
Another issue is that, in this work we only show that the hybrid algorithm satisfies the DA condition, and a more comprehensive study of the ergodicity property
of the algorithm is certainly needed.  
We plan to address these issues in future studies.

\bibliographystyle{siam}
\bibliography{hybrid}
\end{document}